\documentclass[12pt,onecolumn,draftcls]{IEEEtran}
\usepackage{color}
\usepackage{float}
\usepackage{url}
\usepackage{amsmath}
\usepackage{amssymb}
\usepackage{amsfonts}
\usepackage{epsfig}
\usepackage{graphicx}
\usepackage{theorem}
\usepackage{cite}
\usepackage{euscript}
\usepackage{pifont}
\usepackage{fancyhdr}
\usepackage{subfigure}
\usepackage{setspace}

\newtheorem{definition}{Definition}
\newtheorem{proposition}{Proposition}
\newtheorem{corollary}{Corollary}
\title{Low-Density Graph Codes for  Coded Cooperation on Slow Fading Relay Channels}
\author{Dieter~Duyck,~Joseph~J.~Boutros,~and~Marc~Moeneclaey
\thanks{\noindent Dieter Duyck and Marc Moeneclaey wish to acknowledge the activity of the Network of Excellence in Wireless COMmunications NEWCOM++ of the European Commission (contract no. 216715) that motivated this work. The work of Joseph Boutros and part of the work of Dieter Duyck were supported by the Broadband Communications Systems project funded by Qatar Telecom (Qtel)}%
\thanks{Dieter  Duyck  and  Marc Moeneclaey  are  with  the Department  of Telecommunications  and  Information processing,  Ghent
University,   St-Pietersnieuwstraat    41,   B-9000   Gent,   Belgium, \{dduyck,mm\}@telin.ugent.be.}%
\thanks{Joseph J.  Boutros is  with  Texas A\&M University  at Qatar,  PO  Box 23874  Doha, Qatar,  boutros@tamu.edu}
\thanks{\copyright 2009 IEEE. Personal use of this material is permitted. Permission from IEEE must be obtained for all other uses, in any current or future media, including reprinting/republishing this material for advertising or promotional purposes, creating new collective works, for resale or redistribution to servers or lists, or reuse of any copyrighted component of this work in other works.}}

\begin{document}
\maketitle
%%------------------------------------------------------------------
%%------------------------------------------------------------------
\begin{abstract}
We study Low-Density Parity-Check (LDPC) codes with iterative decoding on  block-fading (BF)  Relay Channels.   We consider two users that employ  coded  cooperation, a variant of  decode-and-forward with a smaller  outage probability than the   latter. An outage probability analysis for discrete constellations shows that full diversity can be achieved only when the coding rate does not exceed a maximum value that depends on the level of cooperation. We derive a new code structure by extending the previously published full-diversity root-LDPC code, designed for  the BF point-to-point channel, to exhibit   a   rate-compatibility   property which is   necessary   for   coded cooperation.  We estimate the asymptotic performance through a new density evolution analysis and the word  error rate performance  is determined for finite  length codes.  We show  that our  code  construction exhibits near-outage  limit performance for  all block  lengths and  for a range of coding rates  up to 0.5,  which   is  the highest possible coding rate for two cooperating users.
\end{abstract}

\begin{keywords}
Block fading channels, density evolution, low-density parity-check code, mutual information, relay channels.
\end{keywords}

%%------------------------------------------------------------------
%%------------------------------------------------------------------
\section{Introduction}
When  communicating  over  fading  channels,  Word  Error  Rate  (WER)
performances  as  well  as  power savings  are  dramatically  improved
through  transmit diversity, i.e.,  transmitting signals  carrying the
same  information over different  paths in  time, frequency  or space.
Recently,  a new  network
protocol called \textit{Cooperative Communication} \cite{cover1979ctr,
sendonaris2003ucd,         sendonaris2003ucd2,        laneman2004cdw,
nosratinia2004ccw}  yields  transmit  diversity  using  single-antenna
devices  in  a  multi-user  environment  by taking  advantage  of  the
broadcast  nature of  wireless  transmission. \\  
The most  elementary example of a  cooperative network is the relay  channel, introduced by
van der Meulen \cite{meulen1971ttc}. In a relay channel, a relay helps
the source in  transmitting its data to a  destination by relaying the
messages  sent  by the  source  so that  the  received  energy at  the
destination is increased.  This relay  channel can be generalized to a
cooperative Multiple  Access Channel (MAC)\cite{laneman2004cdw}, where
two  users  transmitting  data  to  a  single  receiver  cooperate  by
alternately  being the  relay  for  the other  user,  as indicated  in
Fig. \ref{fig: Cooperative MAC}. Further generalization to more users is possible, but this will not be
discussed   here  for   simplicity. \\

\begin{figure}[!hpt]
\centering
\includegraphics[width = 0.4\textwidth]{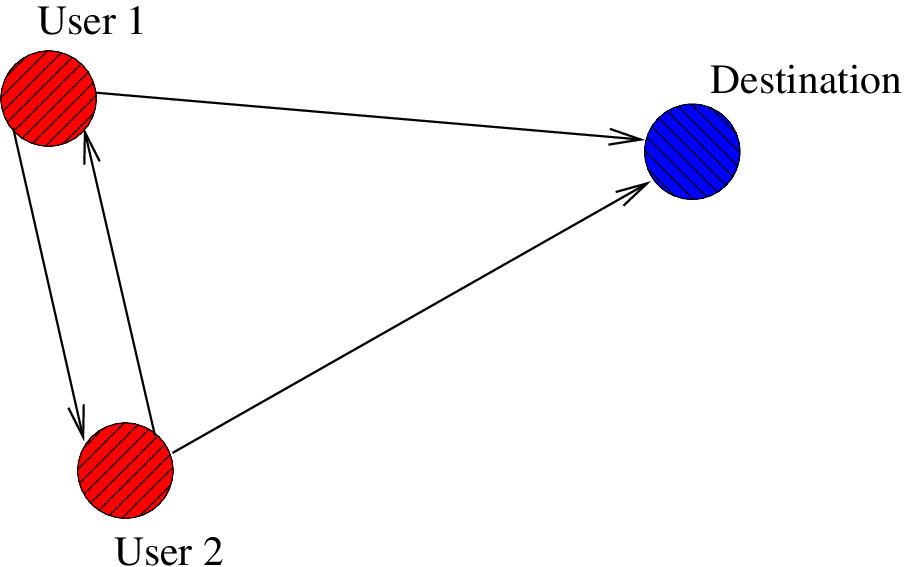}
\caption{A Cooperative Multiple Access Channel (MAC). Arrows between 
two nodes illustrate that both nodes communicate between each other.}
\label{fig: Cooperative MAC}
\end{figure}

A challenging channel  model is the BF\cite{biglieri1998fci} frequency
non-selective  Single-Input  Single-Output  (SISO) channel.  When  the
fading gain  is constant over a  codeword and no  cooperation is used,
the resulting word  error rate curve (displaying the  logarithm of the
error rate versus  the average signal-to-noise ratio (SNR) in dB) has  the same high-SNR slope
as   for    uncoded   transmission:   the    corresponding   diversity
order\footnote{Here, diversity  order is defined  as the ratio  of the
high-SNR slopes of the error  rate curves of the considered system and
of  the uncoded system,  respectively. Alternatively,  diversity order
can be defined as the slope  of the error-rate curve of the considered
system. The diversity  depends on the fading gain  distribution in the
latter definition, but not  in the former definition. Both definitions
are equal in the case of Rayleigh fading.} equals one. The  potential diversity increase  brought by
cooperative techniques allows to save  much transmit energy at a given
error rate. BF channels are  a realistic model
for a  number of channels affected  by slowly varying  fading and flat
fading  is assumed  in  order  to isolate  the  effect of  cooperative
diversity. \\

The  specific task  of  the relay  is  determined by  the strategy  or
protocol. In the case Decode and Forward (DF), the relay
first decodes and then re-encodes the message before sending it to the
destination. A variant of DF is coded cooperation,
where the relay decodes the message received from the source, and then
transmits \textit{additional} parity bits of the message, resulting in
a more spectral efficient strategy \cite{hunter2006dtc}, compared to a
traditional   DF  protocol.   Instead  of   \textit{SNR  accumulation}
(logarithmic rise of mutual information with received power from the relay)   at  the   destination,   we  get   \textit{information
accumulation}  (linear rise of mutual information with received power from the relay)  \cite{zhao2003sna}.  It  has  been  shown  in
\cite{hunter2006oac}       that      the       outage      probability
\cite{biglieri1998fci},  \cite{ozarow1994itc} of coded  cooperation for
half-duplex   BF  channels  is   smaller  than   for  repetition-based
protocols. Moreover, the  concept of coded cooperation can  be used in
more    complex    strategies,    such    as    Amplify-Decode-Forward
\cite{bao2005daf}, where  the relay can  choose between DF and  AF. So
finally, replacing the  decode-and-repeat part in any protocol  by this more
intelligent  ``information  adding''   strategy  improves  the  outage
probability performance. As a consequence, constructing a near-outage channel code for a
coded cooperation  scenario results in  a competitive error-correcting
code in terms of  error-rate performance vs. SNR for
a given rate $R$.\\

Up  till  now, coded  cooperation  has  mainly  been implemented  using
rate-compatible convolutional  codes \cite{hunter2006dtc}. The main drawback of  these codes is that the WER increases
with the logarithm  of block length to the power $d$  where $d$ is the
diversity  order \cite{bou2004tcd}, \cite{bou2005aoc}.  The WER  of practical
near-outage codes should  be independent of the block  length in order
to       approach       the       outage       probability       limit
\cite{fabregas2004ccb}, \cite{fabregas2006cmb}.  The solution is  to use
capacity-achieving codes, for example LDPC codes \cite{richardson2001cld}.  LDPC  codes designed for the special case of
a   cooperative   channel   have   been   reported for   the  Gaussian   channel by   Razaghi and Yu \cite{razaghi2007bld}, \cite{razaghi2007bic} and by 
Chakrabarti et al. \cite{chak2007ldp}.   For  the
block-fading channel however,  there is still a lack  of a near-outage
LDPC code. Hu et al. \cite{hu2007ldp} also designed LDPC codes for the Gaussian relay channel, whereafter they applied this random LDPC code to a BF relay channel. Unfortunately, a random code does not perform very well on a BF relay channel, because it has not the structure to  achieve  full diversity,  as  shown  by  Boutros  et al. \cite{boutros2007daa} and as will  be explained in the rest of the paper.

In Section \ref{sec: outage prob anal}, this paper analyzes the outage probability for binary phase shift keying (BPSK) modulations and derives a coding rate limitation that is necessary for the protocol to have diversity two, valid for all discrete alphabets. Deriving a code structure for coded cooperation will be treated in the second part of the paper. The  aim of coded  cooperation is to  send a codeword  over two independent   fading   paths and the relay must be able to decode after receiving the first part of the codeword. An error-correcting code must therefore exhibit two properties: full-diversity and rate-compatibility. This paper derives a new code structure satisfying both properties. Often \cite{duy2009fjncc, duyck2010aac, duy2011tfdj}, perfect source-relay channels are assumed when designing error-correcting codes. These codes can be extended immediately to codes for cooperative systems with non-perfect source-relay channels using the proposed rate-compatible structure from this paper. 
We also determine density evolution equations to 
%study the performance of an infinite length LDPC code. 
obtain a lower bound on the WER of the LDPC ensemble. The density evolution analysis can also be used to optimize the degree distributions, which will be discussed briefly, but this is not the topic of the paper.\\

Channel-State Information (CSI) is assumed at the decoder. We consider
half-duplex  devices,  assuming   that  simultaneously  receiving  and
transmitting data in the same frequency-band is too complicated due to
the   limited  isolation  of   directional  couplers.   In addition, we also  restrict the protocol to be  orthogonal since we
transmit at low  rates (we use Binary Phase-Shift  Keying (BPSK)). The
proposed code  construction can nevertheless  be used in  more complex
non-orthogonal protocols,  where one can  achieve more coding  gain in
high-rate  scenarios\cite{aza2005ota}.

%%------------------------------------------------------------------
%%------------------------------------------------------------------
\section{System model and notation}
\label{sec: System model and notation}
As  mentioned in  the introduction,  the devices  are  half-duplex and
users transmit  in non-overlapping  time slots. The transmission  of a
codeword is organized in two  frames which constitute one block. We denote the transmission of user $u$, $u=1,2$, in frame $m$, $m=1,2$, by $X_{u,m}$. The pair $(C_{u,1}, C_{u,2})$ denotes the codeword of user $u$. In
the first frame of a block, each user broadcasts the first part of its encoded data to the other
user  and  to the  destination.  In  the  second frame,  users  either
cooperate  or  send  additional parity  bits  related to  their  own  information
message, depending  on whether they are able to decode the transmissions in  the first
frame. The  decoding failure is  detected by the relaying user via a Cyclic  Redundancy Check
(CRC)  code or  any other  intelligent detection  scheme. There  are 4
cases   to  be   distinguished,  as   summarized  in   Fig.  \ref{fig:
4_cases}: in case 1, both users have successfully decoded the information from the other user; in case 2, none of the users has been able to decode the information from the other user; in case 3 (case 4), only user 2 (user 1) has successfully decoded the information from the other user. Methods are known allowing the destination to detect which of these 4 cases has occurred \cite{hunter2004cc}.

\begin{figure*}[!htb]
   \centering
   \subfigure[Case 1. Both interuser transmissions are successfully decoded. Each user cooperates in the second frame.]
   {\includegraphics[width = 0.45\textwidth]{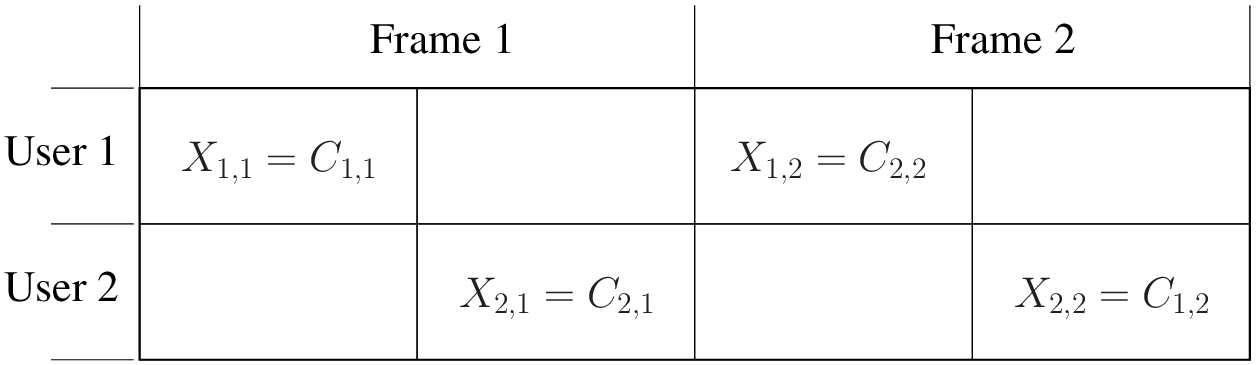}}
   \qquad \subfigure[Case 2. Both interuser communications failed. Each user sends its own parity bits in the second frame.]{\includegraphics[width =0.45\textwidth]
   {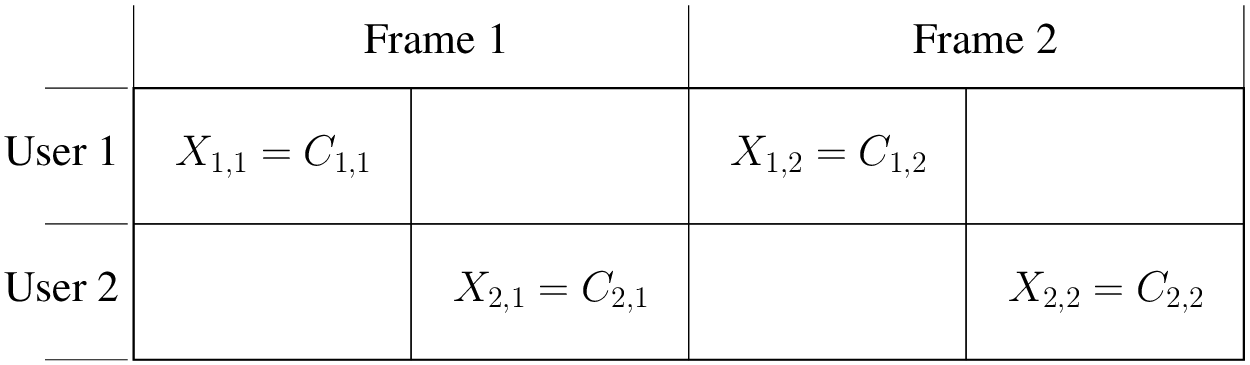}}
	\qquad \subfigure[Case 3. User2-to-User1 communication failed. In the second frame, user 1 sends its own parity bits and user 2 cooperates with user 1.]{\includegraphics[width = 0.45\textwidth]
   {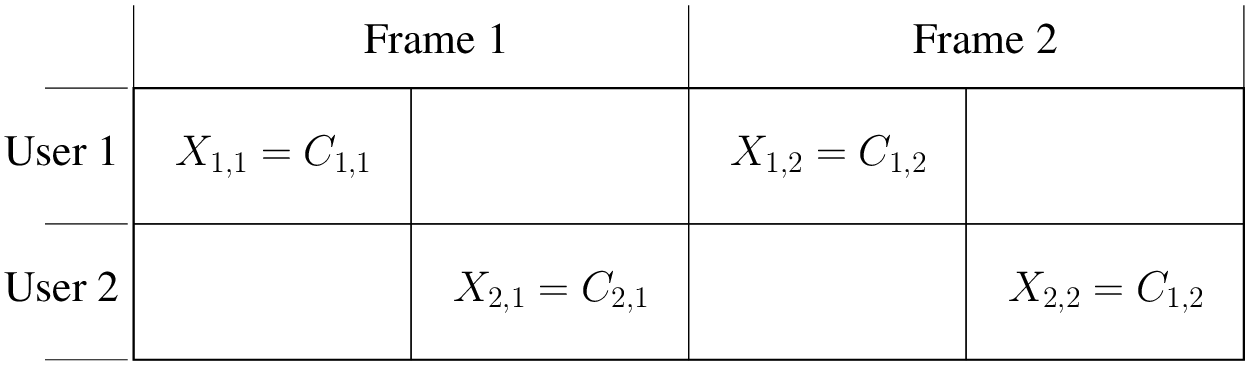}}
	\qquad \subfigure[Case 4. User1-to-User2 communicatino failed. In the second frame, user 2 sends its own parity bits and user 1 cooperates with user 2.]{\includegraphics[width = 0.45\textwidth]
   {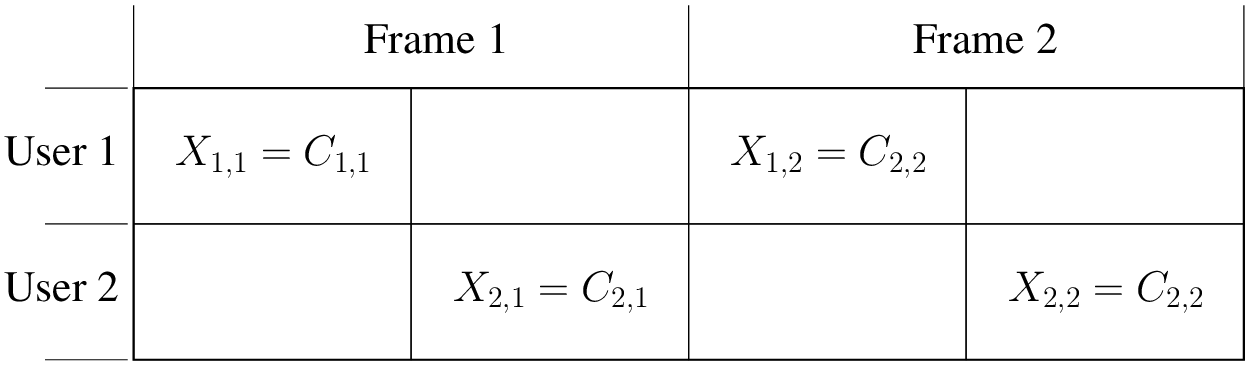}}
		\caption{The 4 cases encountered in coded cooperation are listed above.}
   \label{fig: 4_cases}
\end{figure*}

A codeword will consequently be split over 2 frames. We consider codewords to have a total length equal to $N$ binary digits, where $N=N_1+N_{2}$, and $N_1$ and $N_2$ denote the length of the first and second part of the codeword.
We define the level of cooperation, $\beta$, as the ratio $N_{2}/N$. \\
We denote the transmitter of a frame, which can be user 1 or user 2, by $s$ and the receiver of a frame, which can be user 1, user 2 or the destination, by $r$. Transmitted symbols of user 1 will be denoted $x_{1}[i]$ where $i$ is the symbol time index, $i \in \{1, \ldots, N \}$. Similarly, transmitted symbols
of user 2 are denoted $x_{2}[i]$. The transmitted symbols are chosen from a BPSK alphabet, $x_{s}[i] \in \{ 1, -1\}$. Received symbols will be denoted $y_{sr}[i]$ for received symbols from transmitter $s$ to receiver $r$. The received symbol is given by

\begin{equation}
	\label{eq: channel}
	y_{sr}[i] = \alpha_{sr} x_{s}[i] + z_{r}[i] ,
\end{equation}

\noindent where $z_{r}[i] \sim \mathcal{N}(0,\sigma^2)$ are independent noise samples and $\alpha_{sr} \in \mathbb{R}^+$ is the Rayleigh distributed fading gain between sender $s$ and receiver $r$, with normalized second order moment, $\mathbb{E}[\alpha_{sr}^2]=1$. The fading coefficient $\alpha_{sr}$ is assumed to be constant during 2 frames. Note that this channel model is memoryless \cite{cover2006eit} and satisfies the channel symmetry condition, $p(y_{sr}[i]|\alpha_{sr}, x_{s}[i]=1)=p(-y_{sr}[i]|\alpha_{sr}, x_{s}[i]=-1)$. Each terminal is transmitting at a constant enery per symbol $E_s$, which is related to the energy per information bit $E_b$ by $E_s = R_c E_b$ (BPSK). The total energy per information bit-to-noise ratio is specified by $E_b/N_0$.\\
 
We focus on binary LDPC codes $\mathcal{C}[N,K]_2$ with block length $N$, dimension $K$, and coding rate $R_c = K/N$. Regular LDPC ensembles are characterized by the pair $(d_b, d_c)$, where $d_b$ is the maximum bitnode degree and $d_c$ is the maximum checknode degree. Irregularity is introduced through the standard polynomials $\lambda(x)$ and $\rho(x)$ \cite{richardson2001dca}:
\begin{equation*}
	\lambda(x) = \sum_{i = 2}^{d_{b}}{\lambda_i x^{i-1}} ,~~~~	\rho(x) = \sum_{i=2}^{d_{c}}{\rho_i x^{i-1}}.
\end{equation*}
where $\lambda(x)$ and $\rho(x)$ are the left and right degree distributions from an edge perspective. In Section \ref{sec: density evolution} the polynomials $\mathring{\lambda}(x)$ and $\mathring{\rho}(x)$, which are the left and right distributions from a node perspective, will also be adopted:
\begin{equation*}
	\mathring{\lambda}(x) = \sum_{i = 2}^{d_{b}}{\mathring{\lambda}_i x^{i-1}}, ~~~~ \mathring{\rho}(x) = \sum_{i = 2}^{d_{c}}{\mathring{\rho}_i x^{i-1}}.
\end{equation*}

In this paper, not all bit nodes and check nodes in the Tanner graph will be treated equally. To elucidate the different classes of bit nodes and check nodes, a compact representation of the Tanner graph, adopted from \cite{bou2007dp} and also known as protograph representation \cite{thorpe2003lcp}, will be used. In this compact Tanner graph, bit nodes and check nodes of the same class are merged into one node. \\

\begin{definition}
The diversity order attained by a code $\mathcal{C}$ is defined as \begin{displaymath} d = -\lim_{\gamma \rightarrow \infty} \frac{\log P_e}{\log \gamma}, \end{displaymath} where $P_e$ is the word error rate after decoding.
\end{definition}

\begin{definition} An error-correcting code is said to have \textit{full diversity} if $d=N_u$, where $N_u$ is the number of cooperating users.
\end{definition}

Notice that the above definition assumes Rayleigh distributed single antenna channels. According to the blockwise Singleton bound \cite{knopp2000cbf}, \cite{malkamaki1999epc}, the coding rate for an $n$-order full-diversity code is upper bounded by $R_{cmax} = 1/n$. Hence, in a 2-user scenario we get $R_c \leq 0.5$.

%%------------------------------------------------------------------
%%------------------------------------------------------------------
\section{Outage Probability Analysis}
\label{sec: outage prob anal}

The word error rate of practical systems is, in the limit of large block length, lower bounded by the \textit{information outage probability}
\begin{equation*}
	P_{out} = P\big(I(\alpha, \gamma) < R \big) ,
\end{equation*}
where $I(\alpha, \gamma)$ is the instantaneous mutual information as a function of a certain fading gain $\alpha$ and average SNR $\gamma$, $\gamma=\frac{E_s}{N_0}=\frac{1}{2 \sigma^2}$, where $E_s$ is the symbol energy. This definition remains valid for a channel model as described in (\ref{eq: channel}), but then $\alpha$ is the set of fading gains over a codeword and $\gamma$ is the set of average received SNRs. The rate $R$ is the spectral efficiency of a user, only taking into account its timeslots, hence not the average spectral efficiency\footnote{This is, in our opinion, necessary for a fair comparison between multiple user networks with a different number of users.}. The diversity order of the outage probability limit is the same as the order attained by a full-diversity channel code\cite{fabregas2004ccb}. It is our aim in this paper to approach the outage probability limit for a range of values of the spectral efficiency $R$. Since we use BPSK signaling, the spectral efficiency $R$ is identical to $R_c$.  \\

The outage probability analysis of coded cooperation with a Gaussian alphabet has been made in \cite{hunter2006oac}. Here, the analysis considers BPSK
signaling, leading to an important conclusion in Corollary \ref{cor: outage} at the end of this section. The stated corollary is also valid for larger discrete alphabets.\\
The average mutual information of a SISO channel with received signal $y$, conditioned on the channel realization $\alpha$, is determined by the following well-known formula\cite{ungerboeck1982ccm}:

\begin{equation}
	I(X;Y|\alpha) = 1-\mathbb{E}_{Y|\alpha}\left\{\log_2\left(1+\exp\left[\frac{-2 y \alpha}{\sigma^2}\right]\right)\right\},
	\label{eqn: mut info BPSK}
\end{equation}

\noindent where $\mathbb{E}_{Y|\alpha}$ is the mathematical expectation over $Y$ given $\alpha$. The outage event of a point-to-point link is defined by the mutual information of that link being less than its transmission rate. The outage event $E_o$ of the relay channel is determined by a specific region in the multidimensional space of instantaneous signal-to-noise ratios. Next, we give the exact definition of $E_o$ for coded cooperation with BPSK modulation. We shorten the notation $I(X_i;Y_j|\alpha_{ij})$ to $I_{ij}$.\\

\begin{proposition} \label{prop: outage event for coded cooperation}
	In coded cooperation for a two-user MAC with BPSK signaling, the outage event $E_{o}$ related to user 1 is expressed as follows:
{\footnotesize
\begin{eqnarray*}
  E_{o} &\stackrel{a)}{=}& \left[\left(I_{12} > \frac{R}{1-\beta}\right) \cap  \left(I_{21} > \frac{R}{1-\beta}\right) \cap \left(I_{1d}\left(1\right)<R\right)\right] \nonumber \\
  &\cup& \left[\left(I_{12} < \frac{R}{1-\beta}\right) \cap  \left(I_{21} < \frac{R}{1-\beta}\right) \cap \left(I_{1d}\left(2\right)<R\right)\right] \nonumber \\
  &\cup& \left[\left(I_{12} > \frac{R}{1-\beta}\right) \cap   \left(I_{21} < \frac{R}{1-\beta}\right) \cap \left(I_{1d}\left(3\right)<R\right)\right] \nonumber \\
 &\cup& \left[\left(I_{12} < \frac{R}{1-\beta}\right) \cap   \left(I_{21} > \frac{R}{1-\beta}\right) \cap \left(I_{1d}\left(4\right)<\frac{R}{1-\beta}\right)\right],
\end{eqnarray*}
}where 
{ \small
\begin{eqnarray}
	I_{12} &\stackrel{b)}{=}& 1-\mathbb{E}_{Y|\alpha_{12}} \left\{ \log_2\left( 1+\exp\left[  \frac{-2 y_{12} \alpha_{12}}{\sigma_{12}^2}  \right] \right) \right\} , 
	\label{eq: I12}\\ 
	I_{21} &\stackrel{b)}{=}& 1-\mathbb{E}_{Y|\alpha_{21}} \left\{ \log_2\left( 1+\exp\left[  \frac{-2 y_{21} \alpha_{21}}{\sigma_{21}^2}  \right] \right) \right\},
	\label{eq: I21}
\end{eqnarray}}
and where $I_{1d}(1)$ is $I_{1d}$ in case $i$. For each of the cases considered
in Fig. \ref{fig: 4_cases}, the mutual information $I_{1d}$ can be calculated as follows: 
{\footnotesize
\begin{align}
\intertext{\textbf{Case 1:}}
	\label{eq: I_1d_case1} I_{1d}(1) &\stackrel{c)}{=} 1-(1-\beta) \ \mathbb{E7}_{Y|\alpha_{1d}} \left\{ \log_2\left( 1+\exp\left[  \frac{-2 y_{1d} \alpha_{1d}}{\sigma_{1d}^2}  \right] \right) \right\} 
		\nonumber \\
		& - \beta \ \mathbb{E}_{Y|\alpha_{2d}} \left\{ \log_2\left( 1+\exp\left[  \frac{-2 y_{2d} \alpha_{2d}}{\sigma_{2d}^2}  \right] \right) \right\}. \\
\intertext{\textbf{Case 2:}}
	\label{eq: I_1d_case2} I_{1d}(2) &\stackrel{c)}{=} 1-\mathbb{E}_{Y|\alpha_{1d}} \left\{ \log_2\left( 1+\exp\left[  \frac{-2 y_{1d} \alpha_{1d}}{\sigma_{1d}^2}  \right] \right) \right\}. \\
\intertext{\textbf{Case 3:}}
	\label{eq: out_prob_anal_case3}  I_{1d}(3) &\stackrel{c)}{=} 1- (1-\beta) \ \mathbb{E}_{Y|\alpha_{1d}} \left\{ \log_2\left( 1+\exp\left[  \frac{-2 y_{1d} \alpha_{1d}}{\sigma_{1d}^2}  \right] \right) \right\} \nonumber \\
		& - \beta \ \mathbb{E}_{Y'|{\alpha_{1d}\alpha_{2d}}}\left\{\log_2\left(1+\exp\left[\frac{-2 (y') (\alpha_{1d}^2 + \alpha_{2d}^2)^{3/2}}{\sigma_{1d}^2  \alpha_{1d}^2 + \sigma_{2d}^2 \alpha_{2d}^2}\right]\right)\right\}, \\
	y' &= \frac{(\alpha_{1d}  y_{1d}+\alpha_{2d} y_{2d})}{\sqrt{ \alpha_{1d}^2 + \alpha_{2d}^2}}. \nonumber \\
\intertext{\textbf{Case 4:}}
	\label{eq: I_1d_case4} I_{1d}(4) &\stackrel{c)}{=} 1-\mathbb{E}_{Y|\alpha_{1d}} \left\{ \log_2\left( 1+\exp\left[  \frac{-2 y_{1d} \alpha_{1d}}{\sigma_{1d}^2}  \right] \right) \right\}. 
\end{align}}
\end{proposition}
\vspace{\headheight}
\begin{proof}
\begin{description}
\item[a)] is the union of four events associated to the four cases considered in Fig. \ref{fig: 4_cases}. Each case in $E_{o}$ involves the intersection with an outage event where the mutual information between a user and the destination is below the rate $R$, except for case 4, where only the first frame is dedicated to user 1. 
\item[b)] follows directly from (\ref{eqn: mut info BPSK}). 
\item[c)] uses the fact that the two frames in a block behave as parallel Gaussian channels whose
capacities add together. Of course, both frames timeshare a time-interval, which gives a
weight to each capacity term \cite[Section 9.4]{cover2006eit}, \cite[Section 5.4.4]{tse2005fwc}. 
\item[(\ref{eq: out_prob_anal_case3})] follows from maximum ratio combining \cite{tse2005fwc} at the destination during the second frame. 
\end{description}
\end{proof}

\noindent 
The  outage   probability  is   obtained  by  integrating   the  joint
probability   distribution   $p(\alpha_{12},\alpha_{21},  \alpha_{1d},
\alpha_{2d})$  over the  volume  defined  by $E_o$:
\begin{equation*}
	P_{out} = \iiint_{E_o} p(\alpha_{12},\alpha_{21}, \alpha_{1d},\alpha_{2d}) ~\mathrm{d}\alpha_{12} \mathrm{d}\alpha_{21}\mathrm{d}\alpha_{1d} \mathrm{d}\alpha_{2d}.
\end{equation*}
Just  as for  the Gaussian modulation, there is  only one free parameter $\beta$ because $R$  and   $\gamma$  are  fixed   by  the  protocol  and   the  physical environment. Hence, given  R and $\gamma$, one can  optimize the value of $\beta$. For example, notice  that for a low-SNR interuser channel,
the  outage probability  improves  while taking  $\beta$ smaller  than $0.5$ due to  the enhanced protection of the  source-relay channel. On the  other  hand,  a  $\beta$  smaller than  $0.5$  results  in  lower achievable coding rates, as proved in Corollary \ref{cor: outage}. The optimization of $\beta$, as already undertaken in \cite{hunter2006oac} for Gaussian modulations, is not within the subject of this paper.

\noindent There is an important conclusion to draw from the analysis of Prop. \ref{prop: outage event for coded cooperation}:
\begin{corollary} \label{cor: outage}
In coded cooperation over a block-fading channel for the 2-user MAC with a cooperation level $\beta$, transmitting at a coding rate greater than $\textrm{min}(\beta, (1-\beta))$ renders a single order diversity. 
\end{corollary}
\begin{proof}
A necessary condition for coded cooperation to achieve full diversity over a block-fading channel, is that it achieves full diversity over a Block Erasure  Channel (BEC)\cite{lapidoth1994pcc}, because a BEC  is an  extremal case  of a  block-fading channel. We will show that this condition is not satisfied for coding rates greater than $\textrm{min}(\beta, (1-\beta))$.
In a  BEC, the fading gain $\alpha$  takes two possible values $\{  0, +\infty\}$. An outage event on a point-to-point channel is defined by the fading gain $\alpha$ being zero. As a  consequence, the possible values of the BPSK capacity on a BEC are confined to zero  or one. Hence, for  the two-user MAC,  the mutual information $I_{1d}$ related to case 1  belongs to $\{1, \beta, (1-\beta), 0\}$. A double diversity order is equivalent to stating that two outage events are  necessary to  lose the  transmitted codeword.  Take  the scenario where the  user1-to-destination channel has  fading gain zero  and the user2-to-destination  channel  has   fading  gain  $\infty$.  In  this scenario,  the mutual information  $I_{1d}$ is  equal to  $\beta$. All coding rates higher than $\beta$ will limit the diversity order of the outage probability to one, since  only one channel in outage is enough to lose  the codeword. From a  similar reasoning, it is shown that $R_c$ must be smaller than $(1-\beta)$. This corollary is also valid for
signaling strategies with $M$ constellation points.
\end{proof}
\vspace{\headheight}

In the  sequel, if  not otherwise  stated, we assume  a rate  equal to $R_c=\frac{1}{3}$. From  Corollary \ref{cor:  outage}, we know that the level of cooperation must at least belong to $\beta \in  [\frac{1}{3},~ \frac{2}{3}]$. We  stress on the  fact that the proposed code construction is  very flexible in parameters such as the block length  and the coding rate. We  will use $\beta=0.5$ throughout this  paper,  which allows the broadest range of coding rates according to Corollary \ref{cor: outage}. We illustrate this  in the numerical results by showing the  WER performance  of an  LDPC  code whose  coding rate $R_c$  approaches $1/2$.

%%------------------------------------------------------------------
%%------------------------------------------------------------------
\section{Full-diversity Low-Density coding for coded cooperation}
\label{sec: RC LDPC}
Codewords in coded cooperation are split over 2 frames. The first part of  a codeword,  transmitted  during the  first  frame should  protect information on the noisy source-relay channel. Consequently, a channel
code,  compatible  with  two  distinct  rates is  to  be  devised.  In
non-cooperative   communications,   this    property   is   known   as
rate-compatibility where parity bits of higher rate codes are embedded
in those of lower rate codes \cite{hagenauer1988rcp}. The  advantage is  that all
codes can  be encoded/decoded using a single  encoder/decoder.\\ 
Rate-compatibility in the context of LDPC codes was first
introduced    by   Li    et   al.    \cite{li2002rcl}   and    Ha   et
al.   \cite{ha2004rcp}   and  further   elaborated   for  example   in
\cite{yazdani2004irc}.  Two  techniques   have been  used:  puncturing  and
extending.  A  fraction of  parity  bits of  a  mother  code could  be
punctured  to obtain higher  rate codes.  However, the  resulting rate
range is limited because the deletion  of too many bits has a negative
effect on decoding  via belief propagation.  To obtain  a more dynamic
range  in  rates,  the  technique  of extending  has  been  used.  The
extension  is made  by  adding  extra parity  bits  as illustrated  in
Fig. \ref{fig: parity1}, where the overall code is the intersection of
two constituent codes defined by  $H_2$ and $H_1$ padded with zeros on
the right.

\begin{figure}[!ht]
   \centering
   {\includegraphics[width = 0.35 \textwidth]{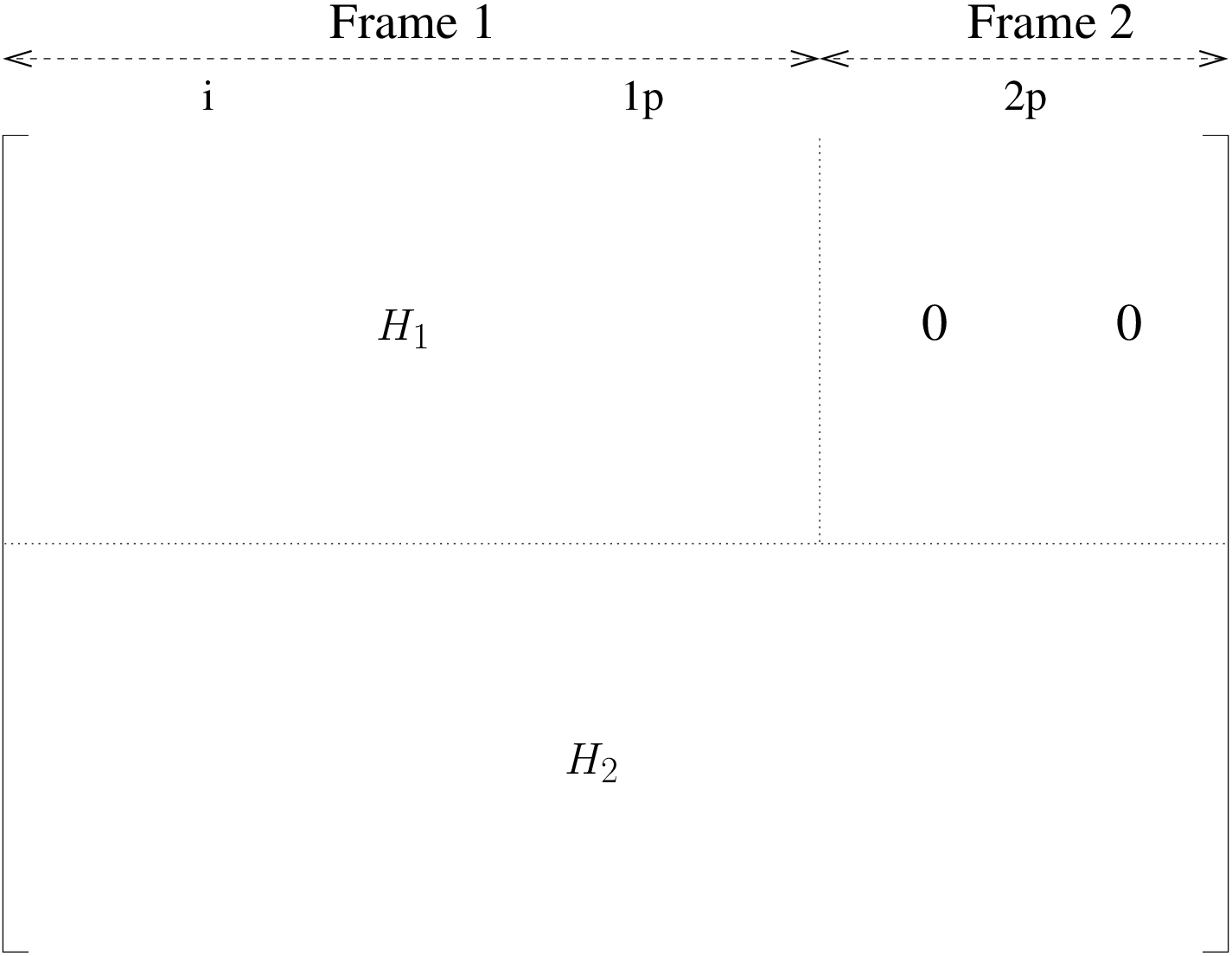}}
   \caption{Parity-check matrix of a rate-compatible LDPC code obtained by the extension of higher rate codes. Symbols are split into three classes: $i$ for the information bits, $1p$ and $2p$ for two classes of parity bits. The classes $i$ and $1p$ are transmitted by the source in frame 1. Parity bits $2p$ are transmitted in the second frame, for example by the relay after successful decoding of the first frame.}
	\label{fig: parity1}
\end{figure}
For simplicity, we only used the technique of extending to acquire rate-compatibility, but this may be further optimized by combining puncturing
and extending via known techniques \cite{li2002rcl, ha2004rcp, yazdani2004irc}.

%%------------------------------------------------------------------
%%------------------------------------------------------------------
\subsection{Full-diversity LDPC codes}
\label{sec: full-diversity}

In coded cooperation, 4 cases occur depending on the success of the transmission
in the first frame. In each of the cases, the destination has other log-likelihood
ratios at the input of the decoder. In the following proposition, we will show
that it is sufficient to guarantee that the decoder at the destination achieves
full diversity in case 1.

\begin{proposition}
	\label{case1: FD}
	In coded cooperation on a cooperative MAC, a code $\mathcal{C}$ attains full diversity, if and only if full diversity is attained in case 1. 
\end{proposition}
\begin{proof}
The WER after decoding $P_e$ can be split as follows
\begin{equation}
	P_e = \sum_{i=1}^4 P(\textrm{case i})P(e| \textrm{case i}).
\end{equation}
The probability that a certain case occurs, depends on the success of decoding two point-to-point channels, so that it is easy to derive that:

\begin{eqnarray}
	P(\textrm{case 1}) &=& (1 - \frac{c}{\gamma})(1 - \frac{c}{\gamma}) \\
	P(\textrm{case 2}) &=& (\frac{c}{\gamma})(\frac{c}{\gamma}) \\	
	P(\textrm{case 3}) &=& (1 - \frac{c}{\gamma})(\frac{c}{\gamma}) \\
	P(\textrm{case 4}) &=& (\frac{c}{\gamma})(1-\frac{c}{\gamma}),
\end{eqnarray}
where $c$ is a positive constant. 
To have $P_e \propto \frac{1}{\gamma^2}$, the following conditions apply:
\begin{eqnarray}
	P(e| \textrm{case 1}) &\propto& \frac{1}{\gamma^2} \label{cond1},\\
	P(e|\textrm{case 2}) &\propto& 1 \label{cond2},\\	
	P(e| \textrm{case 3}) &\propto& \frac{1}{\gamma} \label{cond3},\\
	P(e|\textrm{case 4}) &\propto& \frac{1}{\gamma} \label{cond4}.
\end{eqnarray}
Eqs. (\ref{cond2}), (\ref{cond3}) and (\ref{cond4}) are automatically satisfied, so that the only nessecary and sufficient condition is (\ref{cond1}). 
\end{proof}

Due to Proposition \ref{case1: FD}, we will assume in the following analysis the occurrence of case 1 where the transmission on the interuser channel in the first frame has been successful and both users are cooperating in the second frame. Full-diversity coding on a relay channel must cope with block erasures. Consider the coding  structure plotted in  Fig. \ref{fig:  parity1}. If  all parity bits $2p$ are  erased due to deep fading in frame  2, then the decoder should be capable to retrieve information bits $i$ thanks to $H_1$ and possibly  recompute $2p$  thanks to  $H_2$. Unfortunately,  under deep fading in  frame 1,  a structure with a randomly generated $H_2$, as in Fig.  \ref{fig: parity1}, cannot  guarantee  the  retrieval  of  the  information  bits  through $H_2$. The aim of this section is to explain how $H_2$ can be tuned in order  to   have  full diversity  for   any  left  and   right  degree distribution  and  for  any  block  length.  \\

To the destination, it appears as if one source has sent its codeword over a point-to-point BF channel in case 1. Therefore, we take the constituent code defined by $H_2$ to be a full-diversity LDPC code (referred to as root-LDPC code) as constructed by Boutros et al. in \cite{bou2007dp}, \cite{boutros2007daa} for non-cooperative single-antenna channels with two or more fading states per codeword. The Tanner graph notation for the root-LDPC code is given in Fig. \ref{notation_trellis}. This notation is essential for the analysis because we seek full diversity under iterative decoding.
\begin{figure}[!hbtp]
   \centering
   {\includegraphics[width = 0.35 \textwidth]{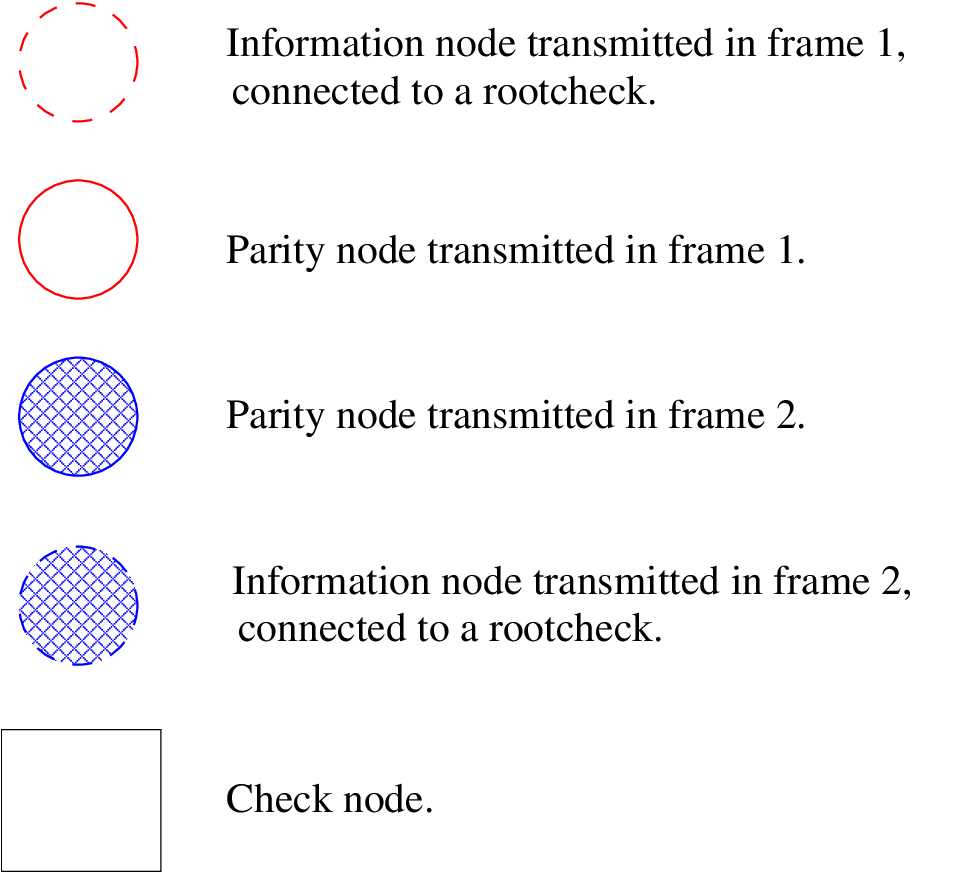}}
   \caption{Notation for the Tanner graph of a full-diversity LDPC code.}
	\label{notation_trellis}
\end{figure}
Full diversity of a root-LDPC structure is created by \textit{rootchecks}, a special type of checknodes in the Tanner graph. As shown in Fig. \ref{fig: rootcheck}, the root and the leaves of this special checknode do not belong to the same frame. When the rootbit is in frame 1, the leavebits are in frame 2, and vice versa. Using the limiting case of a Block-Erasure Channel, it is easy to verify that a rootbit is determined via its rootcheck when its own frame is erased.
\begin{figure}[!hbtp]
   \centering
   \subfigure
   {\includegraphics[width = 0.3\textwidth]{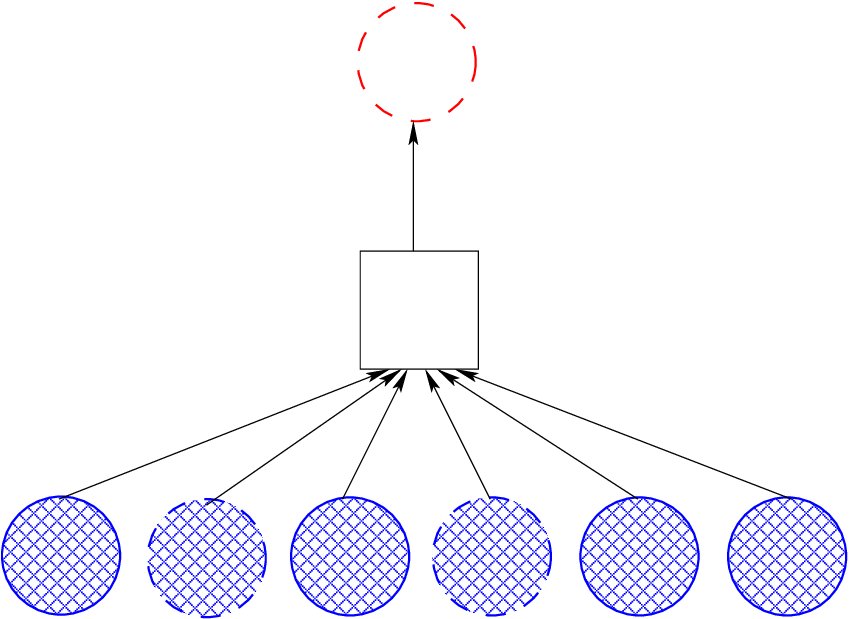}}
   \qquad \subfigure{\includegraphics[width =0.3\textwidth]{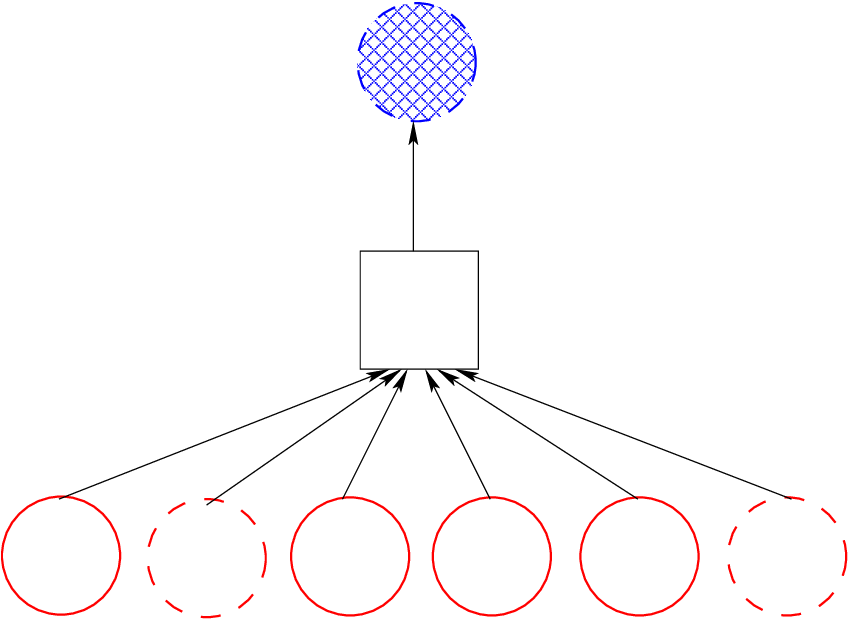}}
   \caption{Two types of rootchecks. On the left-hand side, the rootbit belongs to frame 1 and the leavebits belong to frame 2. The symmetric case where channel states are switched is shown at the right-hand side.}
   \label{fig: rootcheck}
\end{figure}
The complete root-LDPC structure is built after splitting information bits into two classes, denoted $1i$ and $2i$, and parity bits into two classes, denoted $1p$ and $2p$. The checknodes are cut into two classes denoted $3c$ and $4c$\footnote{The checknode notation $1c$ and $2c$ is reserved for $H_1$ in the cooperative code as described in the next subsection.}. The classes $3c$ and $4c$ consist of rootchecks for information bits $1i$ and $2i$ respectively.
The complete root-LDPC structure including all types of nodes is illustrated in Figs. \ref{fig: full-diversity} and \ref{fig: parity-check full-diversity}. Rootchecks are translated into two identity matrices (or permutation matrices in general) inside the parity-check matrix in Fig. \ref{fig: parity-check full-diversity}.

\begin{figure}[!hbtp]
   \centering
   {\includegraphics[width = 0.45 \textwidth]{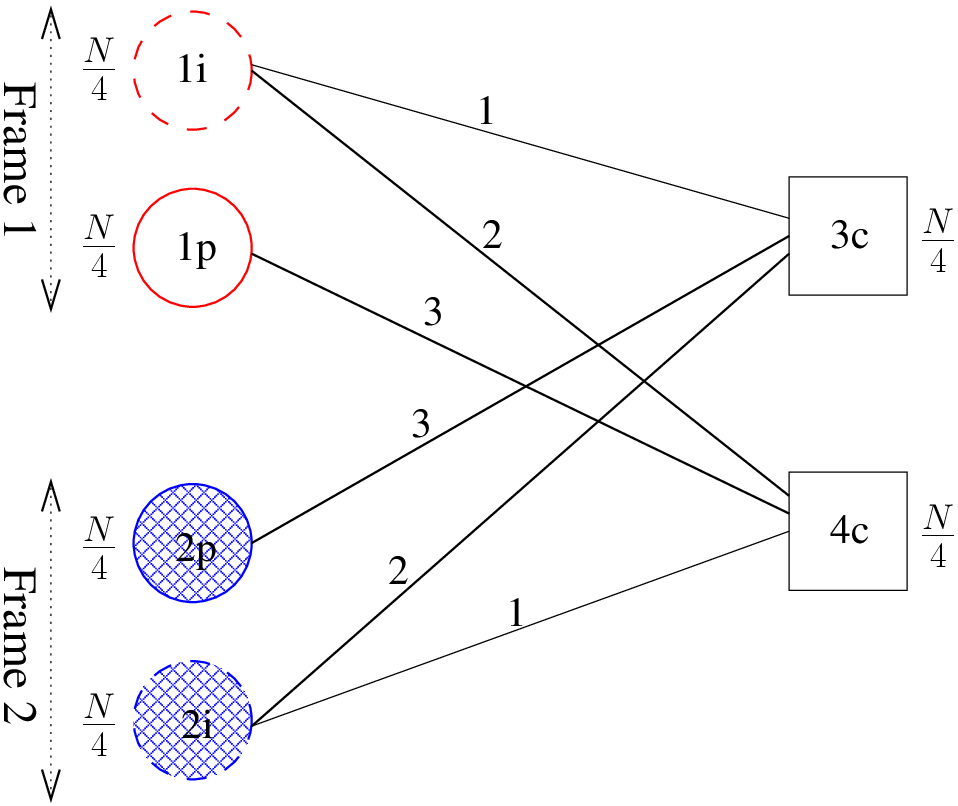}}
   \caption{Tanner graph of a full-diversity LDPC code of length $N$ and rate $\frac{1}{2}$. This compact graph representation has been adopted from \protect\cite{bou2007dp}, \protect\cite{boutros2007daa}, it is also known as protograph representation \protect\cite{thorpe2003lcp}. The integers labeling the edges of the Tanner graph indicate the degree of a node along those edges for a regular (3,6) root-LDPC code. The binary elements are split into four classes of each $\frac{N}{4}$ bits. The checknodes are cut into two classes of $\frac{N}{4}$ checks.}
	\label{fig: full-diversity}
\end{figure}

The proof of full-diversity for block-Rayleigh fading can be found in \cite{bou2007dp}. Note that the diversity order of the root-LDPC code does not
depend on the right or left degree distributions. For simplicity, we only showed a regular (3,6) structure in Fig. \ref{fig: full-diversity}. \\

\begin{figure}[!htp]
   \centering
   {\includegraphics[width = 0.45 \textwidth]{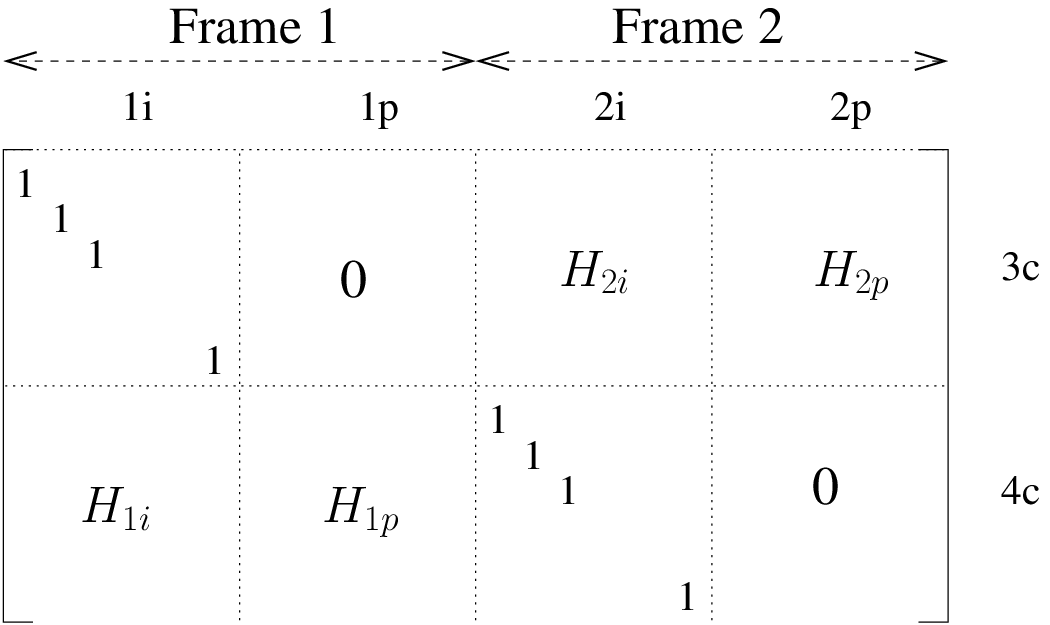}}
   \caption{Parity-check matrix of a rate $\frac{1}{2}$ root-LDPC code. }
   \label{fig: parity-check full-diversity}
\end{figure}

Note that although this code is natural for the point-to-point BF channel, it isn't for the cooperative MAC. The source is sending only half of its information bits to the relay, who is supposed to decode all the information bits. This sounds counter-intuitive and we are the first to apply this concept in cooperative communications. Although it is counter-intuitive, it is necessary to achieve full diversity with iterative decoding, as explained above.

For asymptotic code lengths,  multi-edge type messages propagate in the root-LDPC  graph  \cite{richardson2004met}. One has to choose between two  different root-LDPC   ensembles.   If   we  refer   to  the   Tanner   graph  in Fig. \ref{fig: full-diversity}, the two ensembles are distinguished as follows: (i) The first ensemble is built by two random edge permutations (edge interleavers) connecting $3c$ to ($2i$, $2p$) and $4c$ to ($1i$, $1p$) respectively. This is equivalent to the random generation of two low-density matrices ($H_{2i}$,  $H_{2p}$) and ($H_{1i}$, $H_{1p}$) in the   parity-check  matrix  shown   in  Fig.   \ref{fig:  parity-check full-diversity}. (ii) The  second ensemble is  built by four  random edge permutations $3c-2i$, $3c-2p$, $4c-1i$,  and $4c-1p$. In the root-LDPC parity-check  matrix, this  is equivalent  to building  seperately the four  submatrices  $H_{2i}$, $H_{2p}$,  $H_{1i}$,  and $H_{1p}$.   For simplicity reasons, mainly in  the density evolution (DE) analysis, we adopt  the first  root-LDPC  ensemble as  part  of the  full-diversity
cooperative code proposed in the next subsection.

%%------------------------------------------------------------------
%%------------------------------------------------------------------
\subsection{Rate-compatible full-diversity LDPC codes}
\label{sec: Rate-compatible full-diversity LDPC codes}

The difference with \cite{bou2007dp} is that our code construction must take into account the protocol of coded cooperation, i.e., the 4 different cases, to perform well on this channel. Furthermore, the optimized degree distributions of our code construction will be different from  \cite{bou2007dp}, because of the multi-edge type structure \cite{richardson2004met} of this code construction. The structure of  an LDPC ensemble for coded cooperation is  derived by joining  the   rate-compatibility  property  and   the  full-diversity property. The global parity-check matrix is obtained by  embedding  the   root-LDPC  matrix  (Fig.  \ref{fig:  parity-check full-diversity})  into  the  rate-compatible  matrix  (Fig.  \ref{fig: parity1}). This leads to an  asymmetric code where class $1i$ may have a higher coding gain than class  $2i$. Therefore, we propose an extension to the ``extending'' technique, due to the fact that we split the information bits over two frames, which is a new phenomenon. To get a balanced structure, we replace the zero-padded $H_1$ by the direct sum of  two  rate  $R_1$  codes   defined  by  $H_{1s}$  and  $H_{1r}$  as illustrated in Fig. \ref{fig: rate-compatible full-diversity H}. Thus, the constituent  code $H_{1s}$ protects  bits $1i$ and $1p$  via extra parity bits $p_1'$. Similarly, in  the second frame, extra parity bits $p_2'$  are generated from  $2i$ and  $2p$. The  bottom of  the global parity-check   matrix  simply   includes   the  root-LDPC   structure, connecting ($1i$, $1p$) to ($2i$,  $2p$). For simplicity we can assume that $H_{1s}$ and  $H_{1r}$ belong to the same rate $R_1$ random LDPC ensemble,   defined  by   the  degree   distributions  $(\lambda_1(x), \rho_1(x))$.  Hence, if the  degree distribution  of the  root-LDPC is $(\lambda_2(x), \rho_2(x))$, we refer to the rate-compatible root-LDPC (RCR-LDPC) as  a $(\lambda_1(x),  \rho_1(x),  \lambda_2(x), \rho_2(x))$  code.  The Tanner graphs of  a regular $(3, 9, 3, 6)$ LDPC  code and an irregular $(\lambda_1(x), \rho_1(x), \lambda_2(x), \rho_2(x))$ code are shown in Figs. \ref{fig: Tanner rate-compatible full-diversity H} and \ref{fig: Tanner2   rate-compatible  full-diversity   H}.  Since   we  guarantee full diversity via  a root-LDPC with  a fixed rate  $\frac{1}{2}$, the global coding  rate of the RCR-LDPC  code observed at
the destination is $R_c = \frac{R_1}{2}$. As a consequence, the global coding rate $R_c$ can be easily varied through $R_1$ and is upper limited by $0.5$.

\begin{figure}[!htbp]
   \centering
   {\includegraphics[width = 0.45 \textwidth]{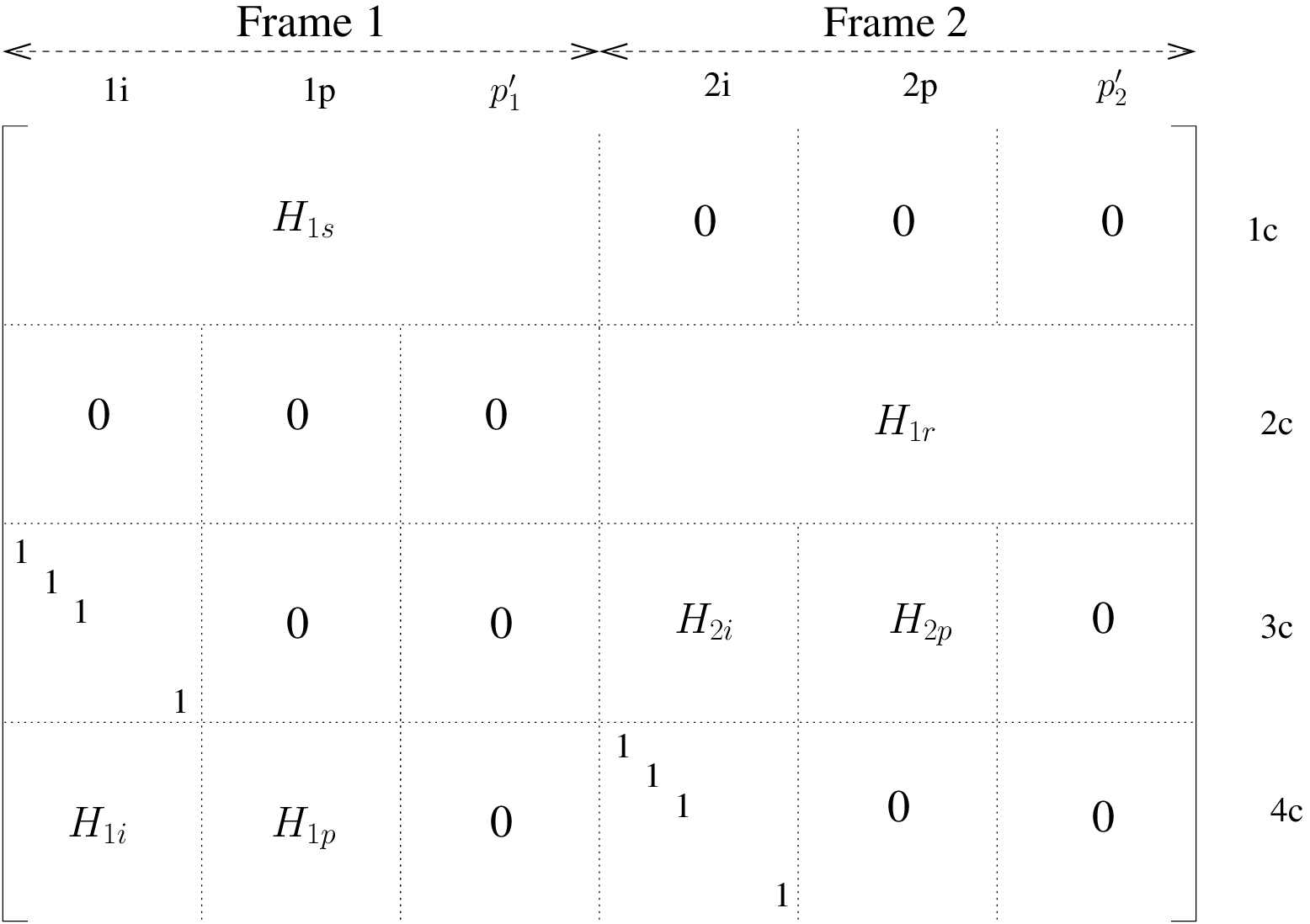}}
   \caption{Parity-check matrix of a RCR-LDPC code for coded cooperation. The upper coding rate associated to $H_{1s}$ and $H_{1r}$ is $R_1 = \frac{2}{3}$, the bottom root-LDPC coding rate is $\frac{1}{2}$, and the overall coding rate is $R_c = \frac{R_1}{2} = \frac{1}{3}$.}
	\label{fig: rate-compatible full-diversity H}
\end{figure}

\begin{figure}[!hbtp]
   \centering
   {\includegraphics[width = 0.45 \textwidth]{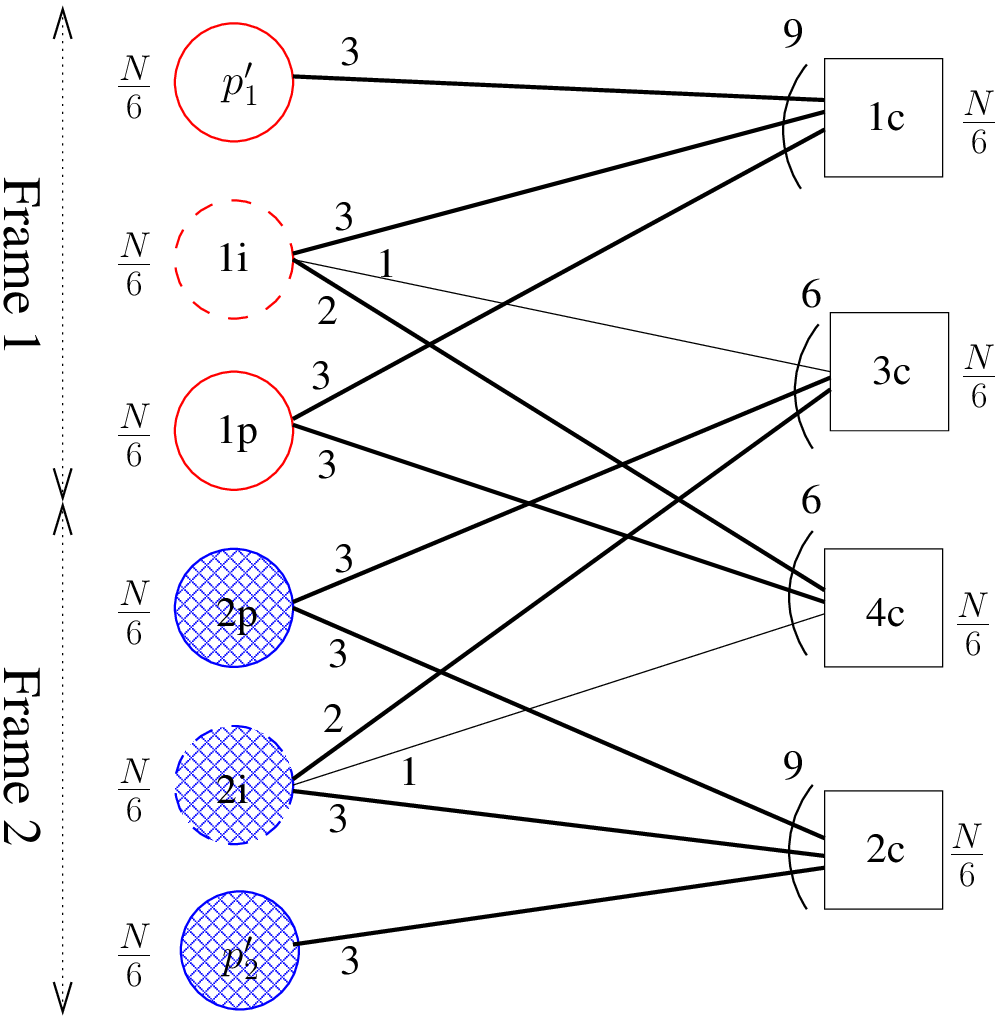}}
   \caption{Tanner graph of a regular $(3,9,3,6)$ RCR-LDPC code for coded cooperation. We see that the average bit degree is $\bar{d_b} = 5$ and the average check degree is $\bar{d_c} = \frac{15}{2}$ which results in $R_c = 1-\frac{\bar{d_b}}{\bar{d_c}} = \frac{1}{3}$.}
	\label{fig: Tanner rate-compatible full-diversity H}
\end{figure}

\begin{figure}[!hbtp]
	\centering
	\includegraphics[width = 0.45 \textwidth]{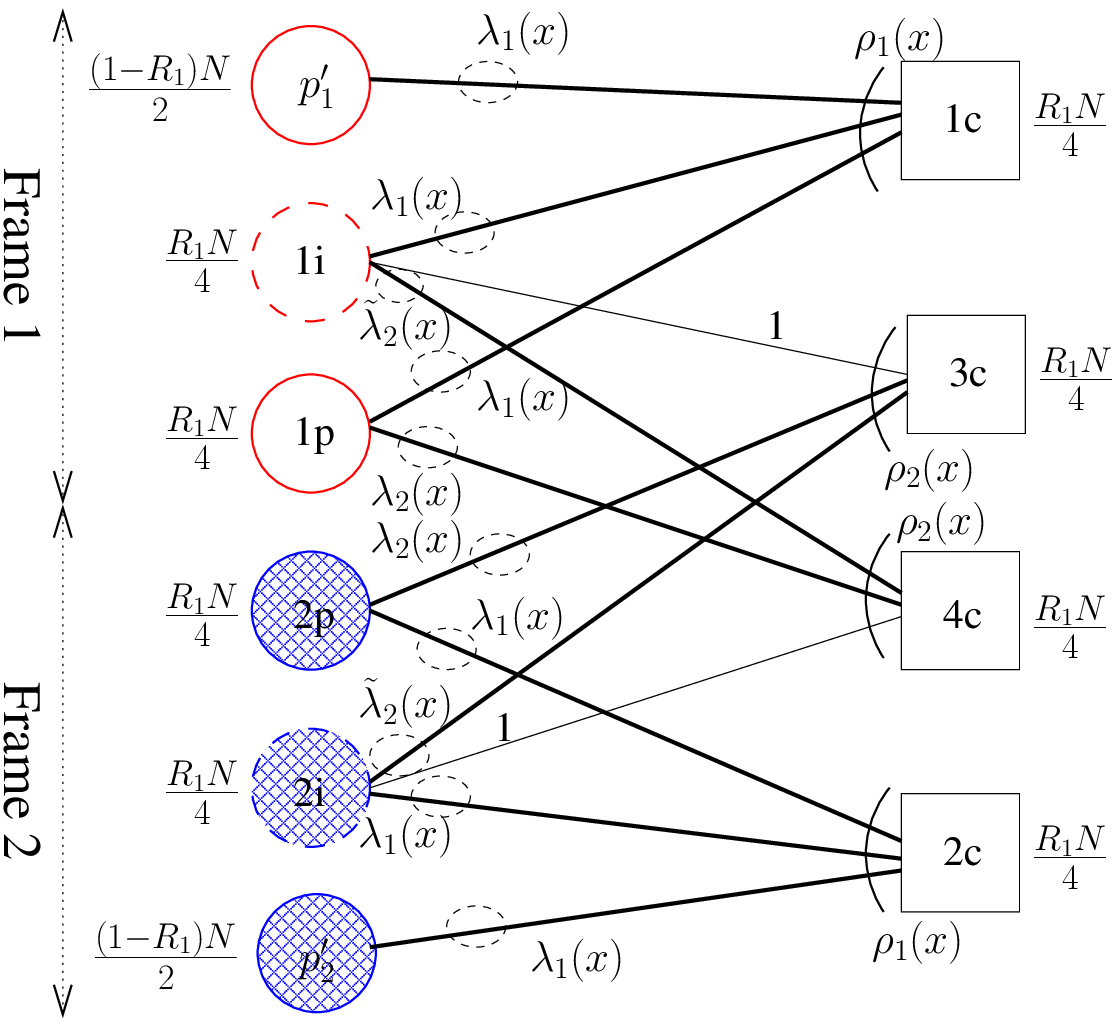}
	\caption{Tanner graph of an irregular RCR-LDPC code for coded cooperation. The binary elements are split into six classes, $p_1'$ and $p_2'$ of each $\frac{(1-R_1)N}{2}$ bits and $1i$, $1p$, $2i$, and $2p$ of each $\frac{R_1N}{4}$ bits. The checknodes are cut into four classes of $\frac{R_1N}{4}$ checks.}
	\label{fig: Tanner2 rate-compatible full-diversity H}
\end{figure}

Due to the identity matrices inside the parity-check matrix, new polynomials $\tilde{\lambda}_2(x)$ appear in Fig. \ref{fig: Tanner2 rate-compatible full-diversity H} in the connections $1i-4c$ and $2i-3c$, as illustrated in Fig. \ref{fig: new polynomial}.

\begin{figure}[!hbtp]
	\centering
	\includegraphics[width = 0.45 \textwidth]{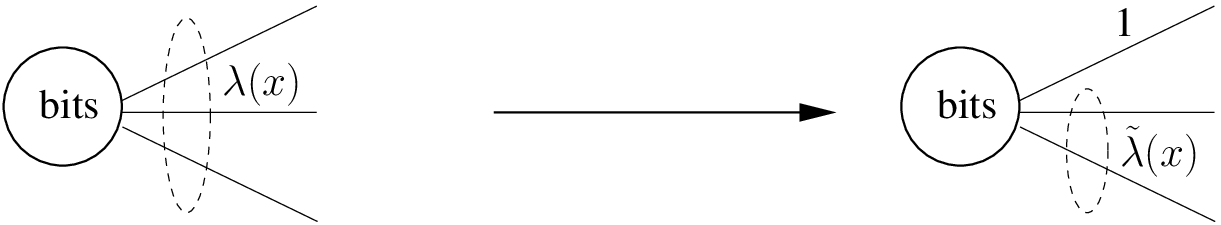}
	\caption{Transition from a traditional representation, characterized by an edge distribution polynomial $\lambda(x)$, towards a representation where one edge per bitnode is isolated resulting in a new degree distribution $\tilde{\lambda}(x)$.}
	\label{fig: new polynomial}
\end{figure}

~\\~\\

\begin{proposition}
\label{prop: newpolynomials}
	In a Tanner graph with a left degree distribution $\lambda(x)$, isolating one edge per bitnode yields a new left degree distribution described by the polynomial $\tilde{\lambda}(x)$:
\begin{equation}
	\tilde{\lambda} \left( x \right) = \sum_i{ \tilde{\lambda}_i \ x^{i-1} }, ~~~~~~ \tilde{\lambda}_{i-1} = \frac{\lambda_i (i-1)/i}{\sum_j{\lambda_j (j-1)/j}}.
\end{equation}
\end{proposition}
\begin{proof}
Let us define $T_{\textrm{bit,}i}$ as the number of edges connected to a bitnode of degree $i$. Similarly, the number
of all edges is denoted $T_{\textrm{bit}}$. From Section \ref{sec: System model and notation}, we know that $\lambda(x)=\sum_{i=2}^{d_{bmax}}\lambda_i x^{i-1}$ expresses the left degree distribution, where $\lambda_i$ is the fraction of all edges in the Tanner graph, connected to a bitnode of degree $i$. So finally $\lambda_i=\frac{T_{\textrm{bit,}i}}{T_{\textrm{bit}}}$.
A similar reasoning can be followed to determine $\tilde{\lambda}_i$:
\begin{eqnarray*}
	\tilde{\lambda}_{i-1} &\stackrel{a)}{=}& \frac{T_{\textrm{bit,}i} - \frac{\lambda_i}{i} T_{\textrm{bit}}}{T_{\textrm{bit}} - \sum_j {\frac{\lambda_j}{j} T_{\textrm{bit}}}} \\
				  &\stackrel{b)}{=}& \frac{\lambda_i T_{\textrm{bit}} - \frac{\lambda_i}{i} T_{\textrm{bit}}}{T_{\textrm{bit}} - \sum_j {\frac{\lambda_j}{j} T_{\textrm{bit}}}} \\
				  &=& \frac{\lambda_i  - \frac{\lambda_i}{i} }{ \sum_j{\frac{\lambda_j}{j} j} - \sum_j { \frac{\lambda_j}{j} }} \\
				  &=& \frac{\frac{\lambda_i}{i} (i-1)}{\sum_j{\frac{\lambda_j}{j} (j-1)}} .
\end{eqnarray*}

\begin{description}
	\item[a)] $\sum_j {\frac{\lambda_j}{j} T_{\textrm{bit}}}$ is equal to the number of edges that are removed which is equal to the number of bits.
	\item[b)] $\lambda_i T_{\textrm{bit}}$ is equal to the number of edges connected to a bit of degree $i$.
\end{description}
\end{proof}
In Section \ref{sec: density evolution}, we will also use $\tilde{\rho}\left( x \right)$, which is defined similarly as $\tilde{\lambda}\left( x \right)$. 

\begin{proposition}
Consider a $(\lambda_1(x), \rho_1(x), \lambda_2(x), \rho_2(x))$ RCR-LDPC code for coded cooperation transmitted on a 2-user block-fading cooperative MAC. Then, under iterative belief propagation decoding, the RCR-LDPC code has full diversity.
\end{proposition}
\begin{proof} 

Let $\Lambda_i^a$, $i = 1 \ldots d_c-1$ denote the input log-ratio probabilistic messages to a checknode $\Phi$ of degree $d_c$.
The output message $\Lambda^e$ for belief propagation is \cite{richardson2008mct}
\begin{equation*}
	\Lambda^e = 2 \textrm{th}^{-1} \left(\prod_{i=1}^{d_c-1} \textrm{th}\left(\frac{\Lambda_i^a}{2}\right)\right),
\end{equation*}

\noindent where $\textrm{th}(x)$ denotes the hyperbolic-tangent function. Superscripts $a$ and $e$ stand for \textit{a priori} and \textit{extrinsic},
respectively. To simplify the proof, we show that the suboptimal min-sum decoder yields a diversity order $2$.
For a min-sum decoder, the output message produced by a checknode $\Phi$ is now

\begin{equation*}
	\Lambda^e = \min{(|\Lambda_i^a|)} \prod_{i=1}^{d_c-1} \textrm{sign}(\Lambda_i^a).
\end{equation*}

\noindent An information bit $\vartheta$ of class $1i$ of degree $d_b$ has $\Lambda_0 = \frac{2 \alpha_{sr} y_{sr}}{\sigma^2}$ where $\Lambda_0$ is the log-likelihood ratio coming from the likelihood $p(y_{sd}|\vartheta)$. It also receives $d_b$ messages: $\Lambda_{1,i}^e$, $i =1 \ldots d_{b1}$ and $\Lambda_{2,i}^e$, $i =1 \ldots d_{b2}$, $d_b = d_{b1}+d_{b2}$, from its neighbouring checknodes in the constituent codes $H_{1s}$ and $H_2$ respectively. The total \textit{a posteriori} message corresponding to $\vartheta$ is $\Lambda = \Lambda_0 + \sum_{i=1}^{d_{b1}} {\Lambda_{1,i}^e} + \sum_{i=1}^{d_{b2}} {\Lambda_{2,i}^e}$. In \cite{bou2007dp} it is proven that full-diversity is achieved if and only if $\Lambda$ behaves as $a \alpha_{1d}^2 + b \alpha_{2d}^2$, where $a, b > 0$. 

The addition of $\sum_{i=1}^{d_{b1}} \Lambda_{1,i}^e$ cannot degrade the error probability $P_e(1i)$ because the convolution with the density of messages from $H_{1s}$ can only physically upgrade the resulting density. Thus, it is sufficient to prove that message $\Lambda_0 + \sum_{i=1}^{d_{b2}} \Lambda_{2,i}^e$ exhibits full diversity, i.e., behaves as $a \alpha_{1d}^2 + b \alpha_{2d}^2$, which is proven in \cite{bou2007dp}.
\end{proof}

%%------------------------------------------------------------------
%%------------------------------------------------------------------
\section{Density Evolution on the Block-Fading Relay Channel}
\label{sec: density evolution}
Richardson  and  Urbanke
\cite{richardson2001cld,  richardson2008mct} established that,  if the
block length  is large  enough, (almost) all  codes in an  ensemble of
codes\footnote{The ensemble of all LDPC-codes that satisfy the left
degree  distribution  $\lambda(x)$   and  right  degree  distribution
$\rho(x)$  is considered.  The  ensemble is  equipped  with a  uniform
probability distribution.}  behave alike, so the  determination of the
average  behavior  is sufficient  to  characterize  a particular  code
behavior. This  average behavior converges  to the cycle-free  case if
the  block  length  augments  and  it can be found in a deterministic  way through density evolution (DE).  The  evolution trees
represent the  local neighborhood of  a bitnode in an  infinite length
code whose graph  has no cycles, hence incoming messages to every node are independent.  \\

%%------------------------------------------------------------------
%%------------------------------------------------------------------
\subsection{Interuser channel}
To determine the density of messages propagating in the graph of the constituent code $H_{1s}$, the following notation is used:

\begin{eqnarray*}
	d_{sr}^m(x) &=& \textrm{density of message from a bitnode to} \\
				&& \textrm{ a checknode in the $\textrm{m}^{\textrm{th}}$ iteration}. \\
	\mu_{sr}(x) &=& \textrm{density of the likelihood of} \\
		&& \textrm{ the source-relay channel.}
\end{eqnarray*}
\noindent Let $X_1 \sim p_1(x)$ and $X_2 \sim p_2(x)$ be two independent real random variables. The density function of $X_1 + X_2$
is obtained by convolving the two original densities, written as $p_1(x) \otimes p_2(x)$. The notation $p(x)^{\otimes n}$ denotes
the convolution of $p(x)$ with itself $n$ times.\\

\noindent Let $X_1 \sim p_1(x)$ and $X_2 \sim p_2(x)$ be two independent real random variables. The density function $p(y)$ of the
variable $Y = 2 \ \textrm{th}^{-1} \left(  \textrm{th} \left( \frac{X_1}{2} \right)  \textrm{th} \left( \frac{X_2}{2} \right) \right)$, obtained through a checknode
with $X_1$ and $X_2$ at the input, is obtained through the \textit{R-convolution} \cite{richardson2008mct}, written as $p_1(x) \odot p_2(x)$.
The notation $p(x)^{\odot n}$ denotes the R-convolution of $p(x)$ with itself $n$ times.\\

\noindent To simplify the notations, we use the following definitions:

\begin{equation*}
	\lambda \left(p \left( x \right) \right) = \sum_i{ \lambda_i \ p(x)^{\otimes i-1} }, ~~  \rho \left(p \left( x \right) \right) = \sum_i{ \rho_i \ p(x)^{\odot i-1} }.
\end{equation*}

\noindent In the next subsection we will also use the following definitions:

\begin{eqnarray*}
	\rho \left(p \left( x \right) , t\left( x \right) \right) &=& \sum_i{ \left( \rho_i \ p(x)^{\odot i-1} \odot t(x)  \right)}, \\
	\lambda^{*} \left(p \left( x \right) \right) &=& \lambda\left(p \left( x \right) \right) \otimes \left(p \left( x \right) \right), \\
	\rho^{*} \left(p \left( x \right) \right) &=& \rho\left(p \left( x \right) \right) \odot \left(p \left( x \right) \right).
\end{eqnarray*}

\noindent The first definition is necessary because of the non-linearity of the R-convolution. Therefore, the first equation is not equal to $t(x) \odot \rho \left(p \left( x \right) \right)$. The next subsection will also use the polynomials $\mathring{\rho}^*\left( x \right)$ and $\mathring{\lambda}^*\left( x \right)$ which are defined by combining the two transformations, denoted by $\mathring{(.)}$ (see introduction) and $(.)^*$. 

Fig. \ref{fig: density subcode} illustrates the local neighborhood of a bitnode in the constituent code $H_{1s}$. 

\begin{figure}[!hbtp]
   \centering
   {\includegraphics[width = 0.48 \textwidth]{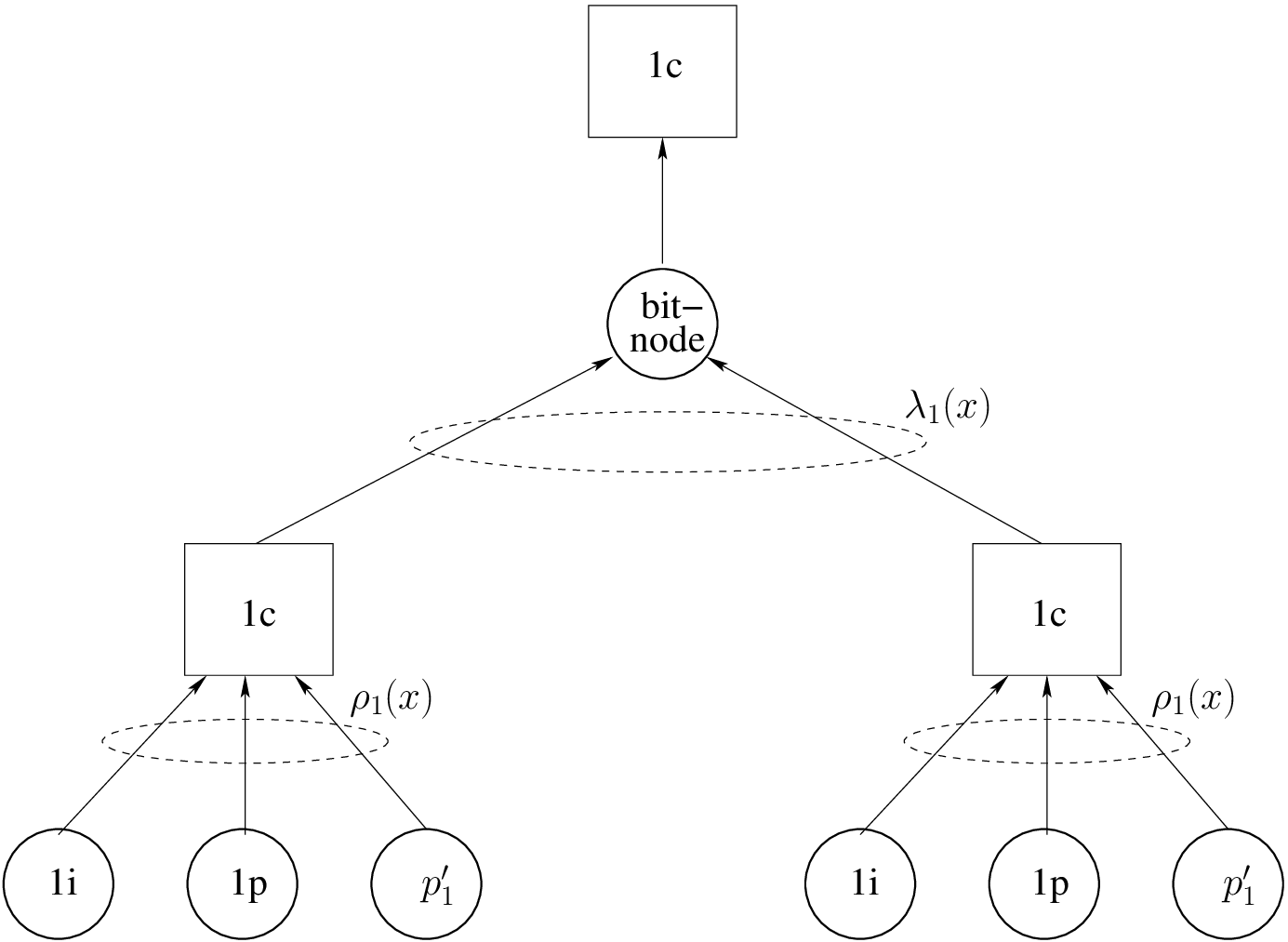}}
   \caption{Local neighborhood of a bitnode in the constituent code $H_{1s}$. This tree is used to determine the evolution of density $d_{sr}(x)$ of messages from a bitnode to a checknode. }
	\label{fig: density subcode}
\end{figure}

\noindent The DE equation in the neighborhood of the bitnode for a $(\lambda_1(x), \rho_1(x))$ LDPC code \cite{richardson2001cld} is, for all $m$,
\begin{equation}
	d_{sr}^{m+1}(x) = \mu_{sr}(x) \otimes \lambda_1\bigg( \rho_1 \big( d_{sr}^{m}(x)\big) \bigg) .
\end{equation}

\noindent The \textit{threshold} of a code is the minimum SNR at which a codeword can be decoded perfectly \cite{richardson2001cld}. Comparing the received signal-to-noise ratio with this threshold, the relay and the source can determine whether the interuser transmissions can be decoded successfully and consequently decide what to transmit in the second frame.\\

%%------------------------------------------------------------------
%%------------------------------------------------------------------
\subsection{Overall cooperative MAC}
The proposed $(\lambda_1(x), \rho_1(x), \lambda_2(x), \rho_2(x))$ root-LDPC code has 6 variable node types and 4 checknode types. Consequently, the evolution of message densities under iterative decoding has to be described through multiple evolution trees. Figs. \ref{fig: density 1i1c_symmetric}, \ref{fig: density 1i3c_symmetric} and \ref{fig: density 1i4c_symmetric} show the local neighborhood of a bit node of the class $1i$. The local neighborhoods of bit nodes of the classes $1p$, and $p_1'$ can be derived similarly. The local neighborhood of classes $2i$, $2p$, and $p_2'$ are equivalent because of code symmetry.  \\

To determine the density of messages, the following notation is used:

\begin{eqnarray*}
	a_1^m(x)\textrm{, }a_2^m(x) &=& \textrm{ density of message from $1i$ to $1c$ and} \\
								&& \textrm{ $2i$ to $2c$ respectively, at the $\textrm{m}^{\textrm{th}}$ iteration}, \\
	f_1^m(x)\textrm{, }f_2^m(x) &=& \textrm{ density of message from $1i$ to $3c$ and} \\
								&& \textrm{ $2i$ to $4c$ respectively, at the $\textrm{m}^{\textrm{th}}$ iteration}, \\
	g_1^m(x)\textrm{, }g_2^m(x) &=& \textrm{ density of message from $1i$ to $4c$ and} \\
								&& \textrm{ $2i$ to $3c$ respectively, at the $\textrm{m}^{\textrm{th}}$ iteration}, \\
	k_1^m(x)\textrm{, }k_2^m(x) &=& \textrm{ density of message from $1p$ to $1c$ and} \\
								&& \textrm{ $2p$ to $2c$ respectively, at the $\textrm{m}^{\textrm{th}}$ iteration}, \\
	l_1^m(x)\textrm{, }l_2^m(x) &=& \textrm{ density of message from $1p$ to $4c$ and} \\
								&& \textrm{ $2p$ to $3c$ respectively, at the $\textrm{m}^{\textrm{th}}$ iteration}, \\
	q_1^m(x)\textrm{, }q_2^m(x) &=& \textrm{ density of message from $p'_1$ to $1c$ and} \\
								&& \textrm{ $p'_2$ to $2c$ respectively in the $\textrm{m}^{\textrm{th}}$ iteration}, \\
	\mu_i(x) &=& \textrm{density of the likelihood of the channel} \\
			&& \textrm{ in the $i$'th frame}.
\end{eqnarray*}

\noindent Note that $\mu_2(x)$ depends on the success or the failure of the transmissions in the first frame. 

\begin{figure}
   \centering
   \resizebox{0.48 \textwidth}{!}{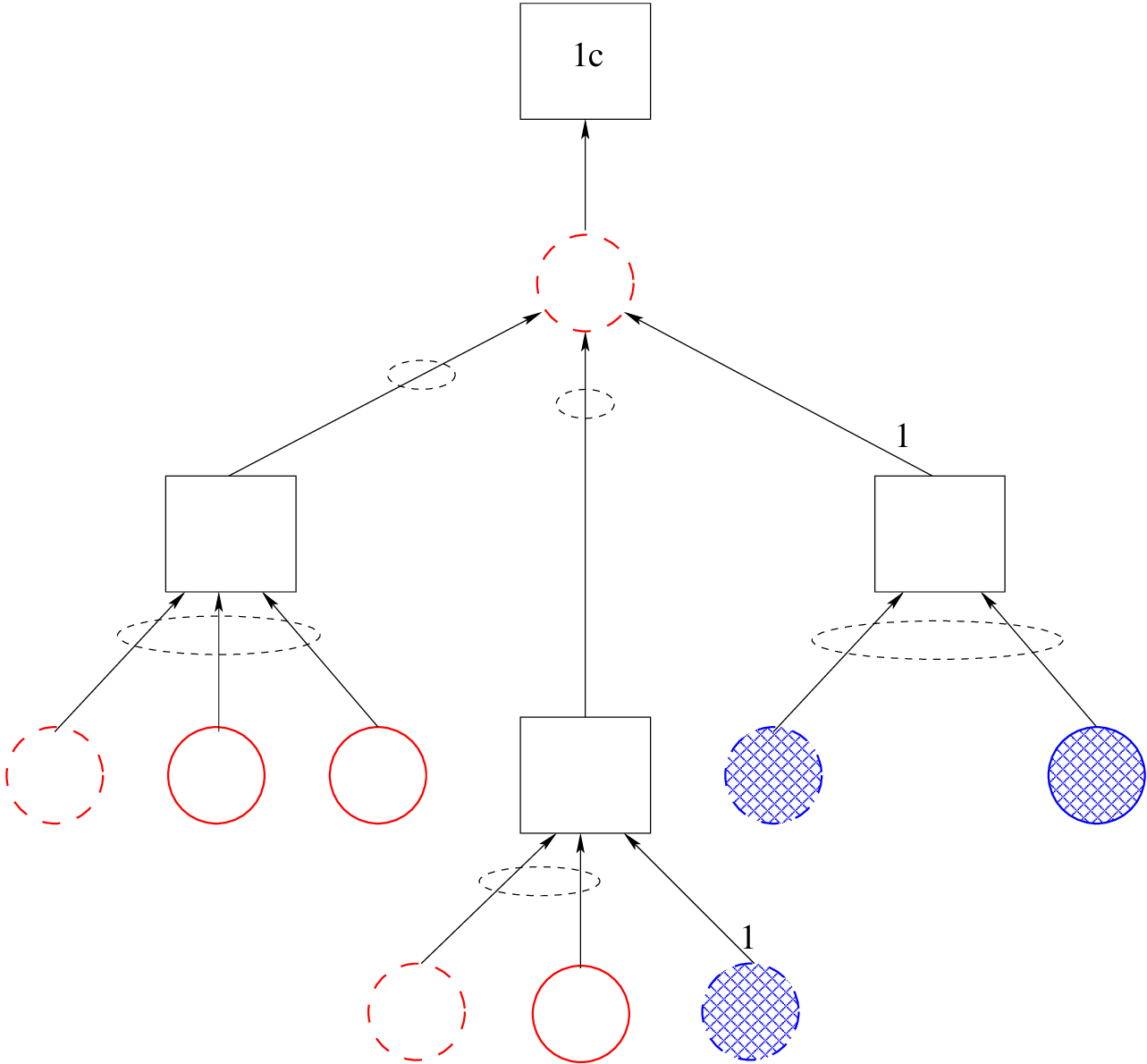}
   \caption{Local neighborhood of bitnode $1i$. This tree is used to determine the evolution of the density of messages $1i \rightarrow 1c$. }
	\label{fig: density 1i1c_symmetric}
\end{figure}

\begin{proposition}
The DE equations in the neighborhood of $1i$ for a $(\lambda_1(x), \rho_1(x), \lambda_2(x), \rho_2(x))$ RCR-LDPC ensemble for coded cooperation, for all $m$, are given in Eqs. (\ref{eq: DE1}), (\ref{eq: DE2}) and (\ref{eq: DE3})

\begin{figure*}
\small
\begin{eqnarray}
	a_1^{m+1}(x) &=& \mu_1(x) \otimes \mathring{\lambda}_2 \bigg(\tilde{\rho_2} \big( f_{1i4c} \ g_1^m(x) + f_{1p4c} \ l_1^m(x), \ f_2^m(x) \big) \bigg) \otimes \lambda_1\bigg(\rho_1 \big( f_{1i1c} \ a_1^m(x) + f_{1p1c} \ k_1^m(x) + f_{p_1'1c} \ q_1^m(x)\big) \bigg) \nonumber \\
 	&& \otimes \mathring{\rho}_2\Big( f_{2i3c} \ g_2^m(x) + f_{2p3c} \ l_2^m(x) \Big), \label{eq: DE1} \\
	f_1^{m+1}(x) &=& \mu_1(x) \otimes \mathring{\lambda}_1^{*}\bigg( \rho_1 \big( f_{1i1c} \ a_1^m(x) + f_{1p1c} \ k_1^m(x) + f_{p_1'1c} \ q_1^m(x)\big)\bigg) \otimes \mathring{\lambda}_2 \bigg( \tilde{\rho_2} \big( f_{1i4c} \ g_1^m(x) + f_{1p4c} \ l_1^m(x) , \ f_2^m(x)\big)\bigg), \label{eq: DE2} \\
	g_1^{m+1}(x) &=& \mu_1(x) \otimes \mathring{\lambda}_1^{*}\bigg( \rho_1 \big( f_{1i1c} \ a_1^m(x) + f_{1p1c} \ k_1^m(x) + f_{p_1'1c} \ q_1^m(x)\big) \bigg) \otimes \tilde{\lambda}_2 \bigg( \tilde{\rho_2} \big( f_{1i4c} \ g_1^m(x) + f_{1p4c} \ l_1^m(x), \ f_2^m(x)  \big) \bigg) \nonumber \\
	&& \otimes \mathring{\rho}_2\Big( f_{2i3c} \ g_2^m(x) + f_{2p3c} \ l_2^m(x) \Big), \label{eq: DE3}
\end{eqnarray}
\end{figure*}

\noindent where
\begin{eqnarray}
	f_{1p4c} &=& \frac{\sum_i \tilde{\rho}_{2i} /i}{\sum_i \lambda_{2i} /i}, \label{eq: DE8}  \\
	f_{1i4c} &=& \frac{\sum_i \tilde{\rho}_{2i} /i}{\sum_i \tilde{\lambda}_{2i} /i}, \label{eq: DE9} \\
	f_{1p1c} &=& \frac{\sum_i \rho_{1i} /i}{\sum_i \lambda_{1i} /i}, \label{eq: DE10} \\
	f_{1i1c} &=& f_{1p1c}, \label{eq: DE11} \\
	f_{p_1'1c} &=& 1 - f_{1i1c} - f_{1p1c}, \label{eq: DE12}  \\
	f_{2i3c} &=& f_{1i4c}, \label{eq: DE13} \\
	f_{2p3c} &=& f_{1p4c}. \label{eq: DE14} 
\end{eqnarray}

\label{prop: density evol 1i}
\end{proposition}

\begin{proof}
Equations (\ref{eq: DE1})-(\ref{eq: DE14}) are directly derived from the local neighborhood trees. To obtain the proportionality factors (\ref{eq: DE8})-(\ref{eq: DE14}), it is important to remark that we use the first ensemble of root-LDPC codes, as explained at the end of Section \ref{sec: full-diversity}. Let T denote the total number of edges between the variable nodes $(1i-1p)$ and the checknodes $4c$. Fig. \ref{fig: prop constants} illustrates how $f_{1p4c}$ and $f_{1i4c}$ are obtained:

\begin{figure}
	\centering
	\includegraphics[width = 0.45 \textwidth]{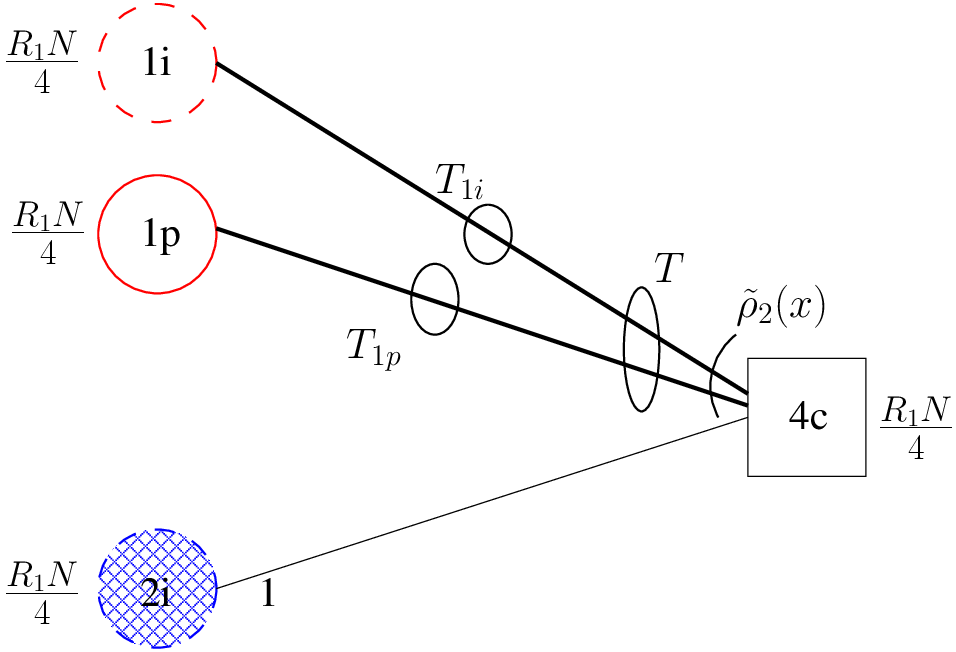}
   \caption{Part of the compact graph representation of the Tanner graph of a root-LDPC for coded cooperation. The number of edges connecting ($1i$, $1p$) to $4c$ is $T$. the number of edges connecting $1p$ to $4c$ is $T_{1p}$. The number of edges connecting $1i$ to $4c$ is $T_{1i}$.}
   \label{fig: prop constants}
\end{figure}

\begin{eqnarray}
	T &\stackrel{a)}{=}& \frac{R_1N/4}{\sum_i \tilde{\rho}_{2i}/i} \\
	T_{1p} &\stackrel{a)}{=}& \frac{R_1N/4}{\sum_i \lambda_i/i} \label{eq: prop T1p}\\
	T_{1i} &\stackrel{a)}{=}& \frac{R_1N/4}{\sum_i \tilde{\lambda}_i/i} \label{eq: prop T1i}\\
	f_{1p4c} &\stackrel{b)}{=}& \frac{T_{1p}}{T} \\
	f_{1i4c} &\stackrel{b)}{=}& \frac{T_{1i}}{T}.
\end{eqnarray}

\begin{description}
\item[a)] The number of checknodes connected to $i$ edges of $T$ is $\frac{\tilde{\rho}_{2i}}{i} T$. A Similar reasoning proves equations (\ref{eq: prop T1p}) and (\ref{eq: prop T1i}).
\item[b)] The fraction of edges $T$ connecting $1p$ to $4c$ is $f_{1p4c}$. The fraction of edges $T$ connecting $1i$ to $4c$ is $f_{1i4c}$. 
\end{description}
\end{proof}

The DE equations in the neighborhood of $1p$ and $p_1'$ for a $(\lambda_1(x), \rho_1(x), \lambda_2(x), \rho_2(x))$ RCR-LDPC ensemble for coded cooperation can be derived similarly. 

\begin{figure}
   \centering
   {\includegraphics[width = 0.48 \textwidth]{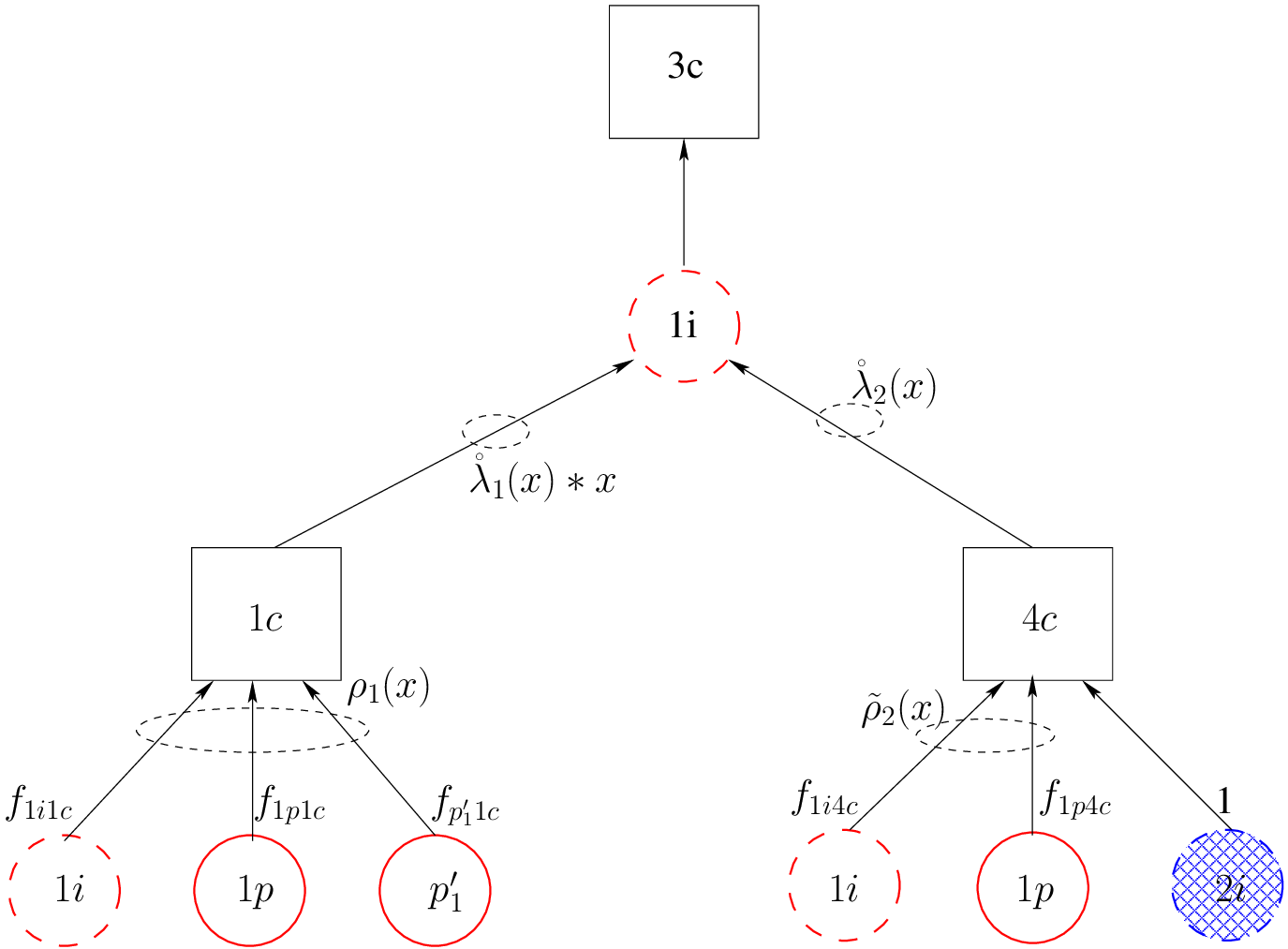}}
   \caption{Local neighborhood of bitnode $1i$. This tree is used to determine the evolution of the density of messages $1i \rightarrow 3c$. }
	\label{fig: density 1i3c_symmetric}
\end{figure}

\begin{figure}
   \centering
	\resizebox{0.48 \textwidth}{!}{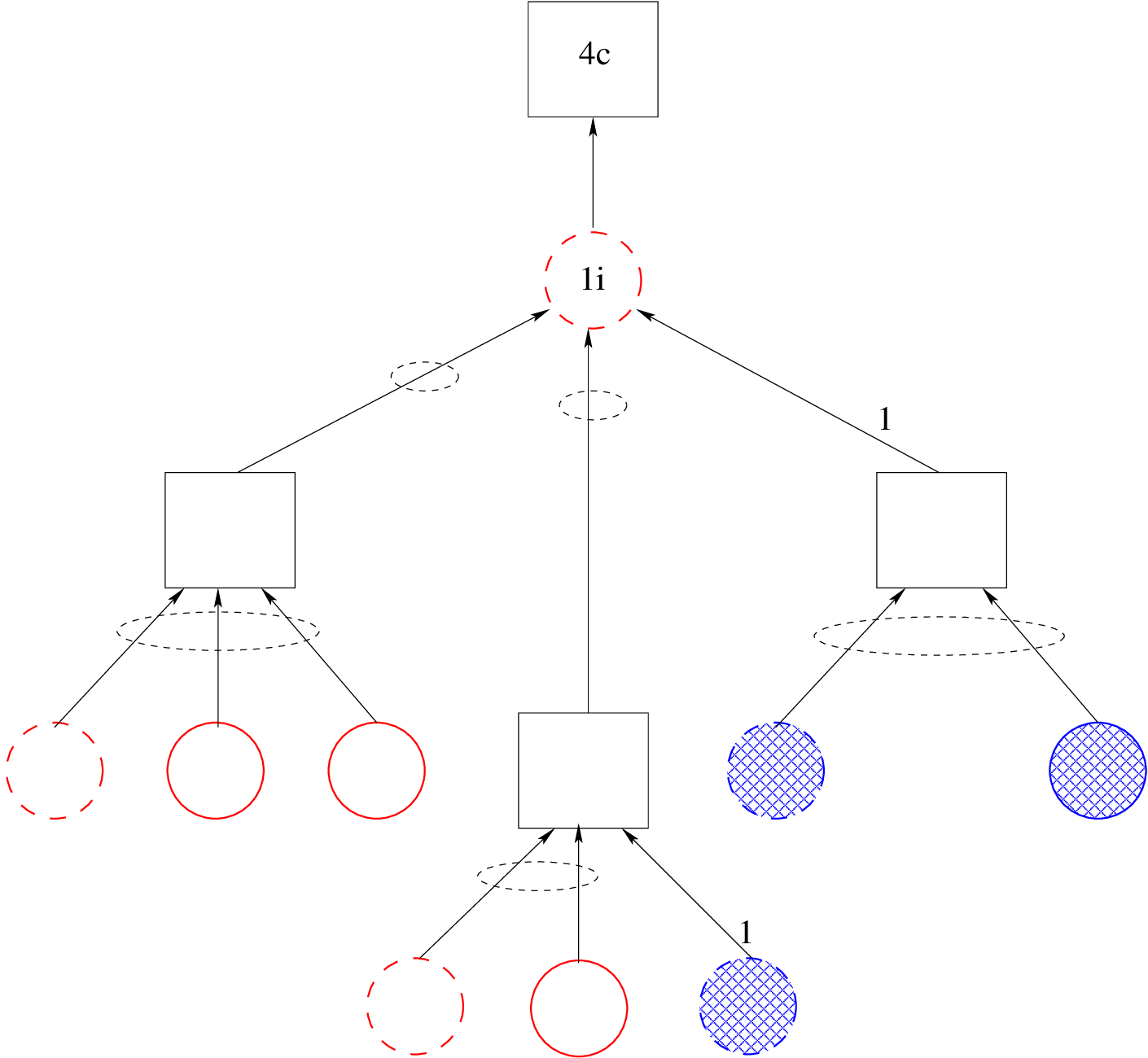}
   \caption{Local neighborhood of bitnode $1i$. This tree is used to determine the evolution of the density of messages $1i \rightarrow 4c$. }
	\label{fig: density 1i4c_symmetric}
\end{figure}

Proposition \ref{prop: density evol 1i} can be used for multiple purposes. First of all, it is used to estimate the asymptotic performance. 
For a fixed fading set $(\alpha_{12},\alpha_{21}, \alpha_{1d}, \alpha_{2d})$, it is possible to determine whether the bit error probability converges to $0$ or not. We refer to the event where the bit error probability does not converge to $0$ by \textit{Density Evolution Outage} ($DEO$). Thus, at a fixed SNR, it is possible to determine the probability of a Density Evolution Outage $P_{DEO}$ by averaging over a sufficient number of fading instances. Now, it is possible to write the word error probability $P_{ew}$ of the ensemble as
\begin{equation}
P_{ew} = P_{ew|DEO} \times P_{DEO} + P_{ew|CONV} \times (1 - P_{DEO});
\label{DEO}
\end{equation}
where $P_{ew|DEO}$ is the word error probability given a \textit{DEO} event, $P_{ew|DEO}=1$, and $P_{ew|CONV}$ is the word error probability when DE converges. The probability $P_{ew|CONV}$ depends on the speed of convergence of density evolution and the population expansion of the ensemble with the number of decoding iterations \cite{Jin2005bei}, so that
\begin{equation}
	P_{DEO} \leq P_{ew}.
	\label{DEO lowerbound}
\end{equation}
Thus, the performance estimated via density evolution is a lower bound for the word error probability.

Secondly, Proposition \ref{prop: density evol 1i} can be used to determine the threshold of $\mathcal{C}$ on an ergodic channel. This does not directly serve the performance analysis for the BF channel. However, an analysis in the real space of the fading coefficients has shown that this can be used to increase the coding gain on a BF relay channel \cite{duyck2010uld}. But the optimization of the coding gain is outside the scope of this paper and here we will only use Proposition \ref{prop: density evol 1i} in the application of Eq. (\ref{DEO lowerbound}).

%%------------------------------------------------------------------
%%------------------------------------------------------------------
\section{Numerical Results}

In this section we estimate the asymptotic performance of RCR-LDPC codes through DE and verify Eq. (\ref{DEO lowerbound}) through finite length simulations. We studied different scenarios:\\

\subsubsection{Scenario 1}
\begin{itemize}
	\item The average SNR of the independent interuser channels is 5dB higher than the average SNR on the source-destination link.
	\item The average SNR of the relay-destination link is equal to that on the source-destination link.
	\item The coding rate is $R_c = \frac{1}{3}$ and the cooperation level is $\beta = 0.5$.
\end{itemize}
For this scenario, we have tested two code ensembles: a regular (3,9,3,6) RCR-LDPC code and an irregular $(\lambda_1(x), \rho_1(x), \lambda_2(x), \rho_2(x))$ RCR-LDPC code with left and right degree distributions given by the polynomials 
{\small
\begin{eqnarray*}
	\lambda_1(x) &=& 0.1989 x + 0.2305 x^2 + 0.0068 x^5 + 0.2774 x^6 \\ 
		&& + 0.14267 x^{19} + 0.1335 x^{20} + 0.0102 x^{21}, \\
	\rho_1(x) &=& x^{12}, \\
	\lambda_2(x) &=& 0.22767 x + 0.20333 x^2 + 0.2145 x^5 \\ 
		&& + 0.011048 x^6 + 0.34346 x^{19}, \\
	\rho_2(x) &=& 0.5 x^7 + 0.5 x^8 .
\end{eqnarray*}}

\subsubsection{Scenario 2}
\begin{itemize}
	\item The average SNR of the independent interuser channels is 12dB higher than the average SNR on the source-destination link.
	\item The average SNR of the relay-destination link is 4dB higher than the average SNR on the source-destination link.
	\item The coding rate is $R_c = 0.45$ and the cooperation level is $\beta = 0.5$.
\end{itemize}
Here, we imitated the channel conditions used in \cite{hu2007ldp}\footnote{We use the same distribution of the fading and the same average SNR. However, in \cite{hu2007ldp}, the source keeps transmitting in the second frame, so that a direct comparison between our code and the performance of the code proposed in \cite{hu2007ldp} is not possible.}. The average SNR of the interuser channels is high with respect to the uplink channels, allowing a high coding-rate for the source-relay channel. We used an irregular $(\lambda_1(x), \rho_1(x), \lambda_2(x), \rho_2(x))$ RCR-LDPC ensemble with left and right degree distributions given by the polynomials 
{\small
\begin{eqnarray*}
	\lambda_1(x) &=& 0.1581 x + 0.2648 x^2 + 0.1116 x^5 + 0.1354 x^6 \\
	&& + 0.3301 x^{14},\\
	\rho_1(x) &=& x^{43} ,\\
	\lambda_2(x) &=& 0.234413 x + 0.21392 x^2 + 0.123711 x^5 + 0.125548 x^6 \\ 
	&& + 0.30241 x^{19}, \\
	\rho_2(x) &=& 0.71875 x^7 + 0.28125 x^8 .
\end{eqnarray*}}
The coding rate for the interuser channel subcode $H_1$ is equal to $0.9$. 

\subsection{Density Evolution Outage}
\label{subsec: infinite length}

We evaluated the asymptotic performance of RCR-LDPC codes by applying DE on the proposed code construction. The probability of Density Evolution Outage $P_{DEO}$, which is a lower bound of the WER, for both scenarios is illustrated in Fig. \ref{fig: rate=0.333, beta=0.5, interuser=5}.
\begin{figure*}[!hbtp]
   \centering
   {\includegraphics[width = 0.6 \textwidth, angle = -90]{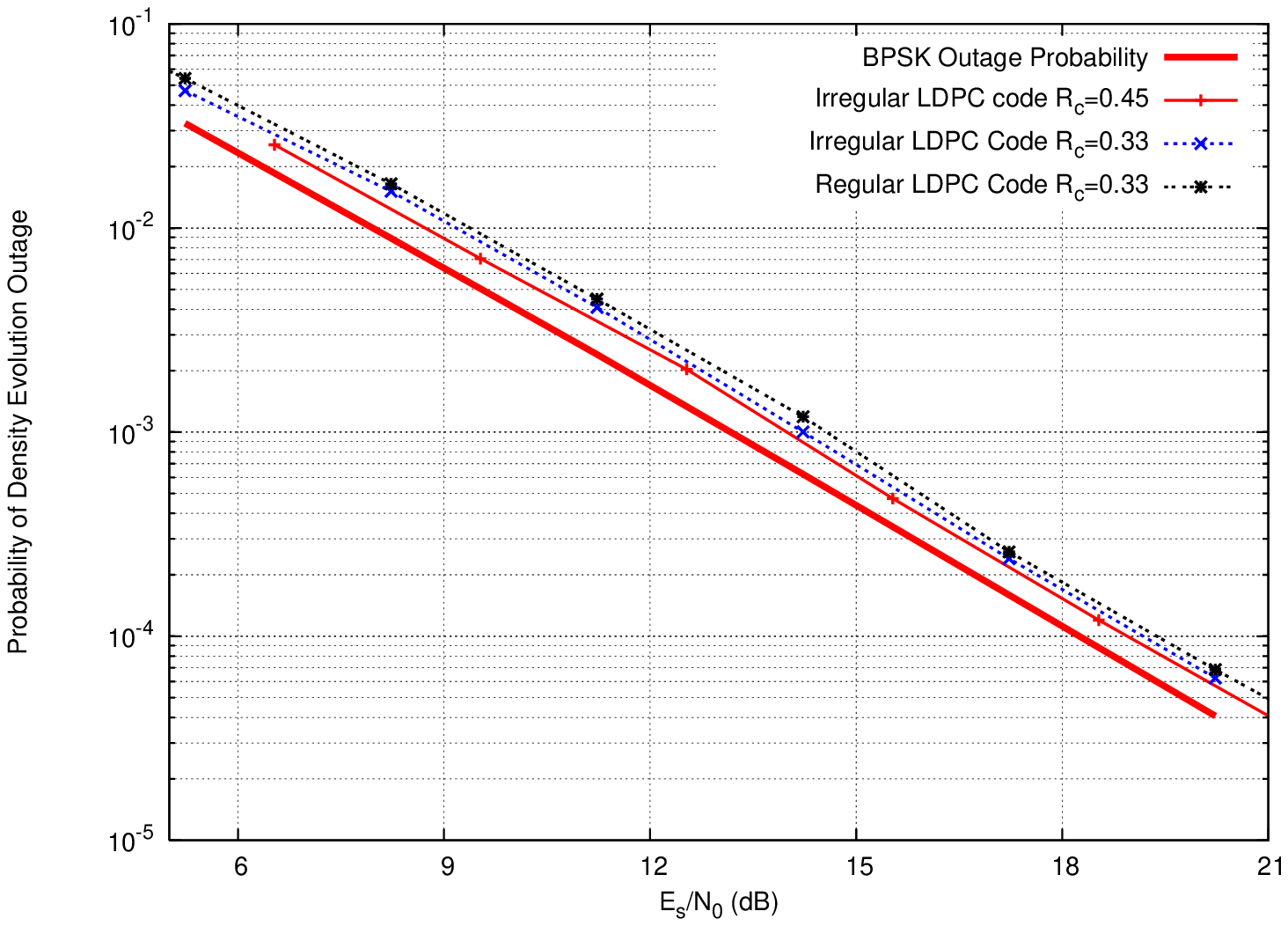}}
   \caption{Density Evolution Outage probability of RCR-LDPC codes with coding rates $R_c = \frac{1}{3}$ (scenario 1) and $R_c=0.45$ (scenario 2) with iterative decoding on a cooperative MAC with two users. $E_s/N_0$ is the average symbol energy-to-noise ratio on the source-destination link.}
	\label{fig: rate=0.333, beta=0.5, interuser=5}
\end{figure*}
Note that the outage probability for both rates is, by coincidence, too close to distinguish. The simulated RCR-LDPC code ensembles all perform within $1.5 \textrm{dB}$ from the outage probability limit, whereas the irregular  RCR-LDPC code ensembles are within $1 \textrm{dB}$ from  the outage probability  limit.  This  distance is respected for many variations of the channel conditions, such as other interuser  channel conditions  or uplink channel conditions. Note that our code construction can be applied on a full-duplex channel, doubling the overall spectral efficiency. As mentioned before, the coding rate is  adjustable by varying the number of parity bits $p_1'$ and $p_2'$, which is illustrated in scenario 2. \\

In this work, we mainly focussed on the diversity order achieved by the code construction. In more recent work \cite{duyck2010uld} we optimized the degree distribution using the analysis of Section \ref{sec: density evolution}. Another method is based on density evolution with a modified Gaussian approximation that takes into account the SNR variation in one received codeword as well as the rate-compatibility constraint \cite{li2008lcd}.

%%------------------------------------------------------------------
%%------------------------------------------------------------------
\subsection{Finite Length LDPC Codes}
It is interesting to evaluate the finite length performance of the proposed RCR-LDPC codes. Not only to approve the asymptotic performance, but also to see how to generate an instance of the parity-check matrix, given by Fig. \ref{fig: rate-compatible full-diversity H}. Before showing the results, we will first discuss the practical generation of this parity-check matrix. 

Consider case 1 from Fig. \ref{fig: 4_cases}. For the decoding process, the destination will apply the sum-product algorithm on the overall graph including $H_{1s}$,  $H_{1r}$, and  $H_2$. For the encoding process, it is easier to determine the parity bits $p_1'$, $p_2'$, and $(1p, 2p)$ with the parity-check matrices  $H_{1s}$,  $H_{1r}$, and  $H_2$ respectively. As with standard LDPC encoding, these matrices will then be systemized to determine the parity bits. An important constraint for the decoding process is the alignment in the overall parity-check matrix of common bit nodes in both constituent codes. This can be achieved by prohibiting column permutations during the systemization of $H_{1s}$,  $H_{1r}$ and  $H_2$. Except for case 4, which only decodes on $H_{1s}$, the other cases need the same constraints. \\

\subsubsection{Generation of $H_{1s}$ and $H_{1r}$}
$H_{1s}$ and $H_{1r}$ are randomly generated satisfying the degree distribution $\rho_1(x)$ for its rows and the degree distribution $\lambda_1(x)$ for its columns. A sufficient condition to prohibit column permutations during the systemization of $H_{1s}$ and $H_{1r}$ is imposing on $H_{p_1'}$ and $H_{p_2'}$ to be full-rank. $H_{p_1'}$  ($H_{p_2'}$ respectively) is  the most  right square  matrix of  $H_{1s}$ ($H_{1r}$  respectively). \\

\subsubsection{Generation of $H_2$}
The generation  of $H_2$ can be  split in the  generation of $H_{4c}$ and  $H_{3c}$,   where  $H_{3c}$   ($H_{4c}$  resp.)  is the upper  part (resp.  lower part)  of  the parity-check  matrix  $H_2$. $H_{3c}$  is the concatenation of an identity matrix (permutation matrix), zeros and a randomly generated matrix $(H_{2i}, H_{2p})$. The rows of $(H_{2i}, H_{2p})$ satisfy the degree distribution $\tilde{\rho}_2(x)$, the columns of the most left square matrix $H_{2i}$ satisfy the degree distribution $\tilde{\lambda}_2(x)$ and the columns of the most right square matrix $H_{2p}$ satisfy the degree distribution $\lambda_2(x)$. This is equivalent to generating a random graph with two classes of bitnodes at the left side and one class of checknodes at the right side of the graph. If $n_{3c}$ is the number of checknodes at the right side, then a random graph with $\frac{n_{3c}}{\sum_i\tilde{\rho}_{2i}}$ edges is generated. A fraction $\frac{\sum_i \tilde{\rho}_{2i}/i}{\sum_i \tilde{\lambda}_{2i}/i}$ of the edges is connected to bit nodes of the class $2i$, whereas a fraction $\frac{\sum_i \tilde{\rho}_{2i}/i}{\sum_i \lambda_{2i}/i}$ of the edges is connected to bit nodes of the class $2p$. In the end, the identity matrix is simply added. $H_{4c}$ is generated similarly. \\
For the encoding process, we have to systemize this matrix. One solution is to switch the columns associated with the $1i$ bit node class and the $2p$ bit node class. The most left square matrix of $H_2$ will then be block-diagonal with $H_{2p}$ and $H_{1p}$ on its diagonal. Having  $H_{2p}$  and   $H_{1p}$  full-rank  is consequently a  sufficient condition to exclude column permutations during the systemization of this matrix. After the generation of $(2p, 1p)$, all the bits are put in the required order $1i-1p-2i-2p$ by switching back the bits of the classes $1i$ and $1p$.\\

\subsubsection{WER performance of finite length LDPC codes}
The probability of Density Evolution Outage $P_{DEO}$ is a lower bound of the WER of LDPC ensembles without cycles in its Tanner graph, which is illustrated in Fig. \ref{fig: irregular finite length} for irregular codes and in Fig. \ref{fig: finite_LDPC_1} for the regular code of scenario 1. In the latter, we augment the blocklength to show that the WER of LDPC codes is independent of the block length. The results shows that inequality (\ref{DEO lowerbound}) is very tight in this case.

\begin{figure}[!hbtp]
   \centering
   {\includegraphics[width = 0.33 \textwidth, angle = -90]{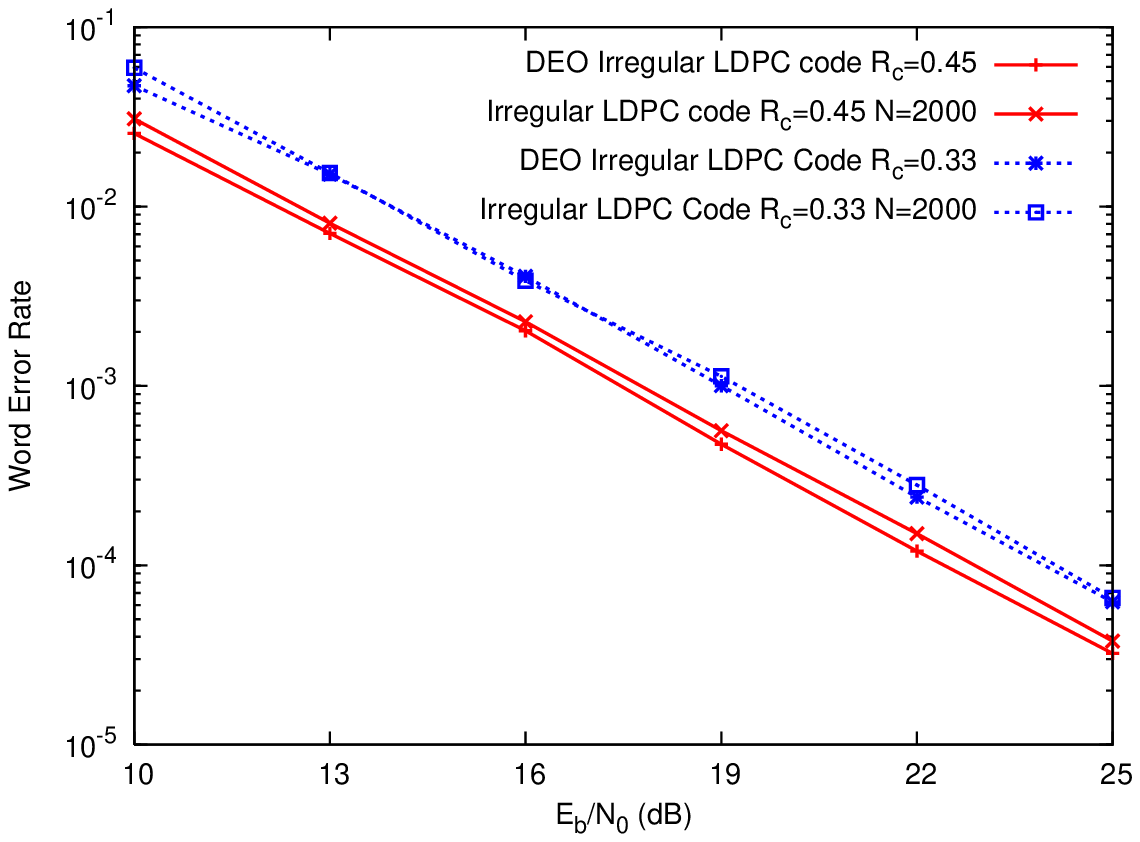}}
   \caption{Comparison of Density Evolution Outage (DEO) probability of irregular RCR-LDPC codes with coding rates $R_c = \frac{1}{3}$ (scenario 1) and $R_c=0.45$ (scenario 2) with iterative decoding on a cooperative MAC with two users. $E_b/N_0$ is the average information bit energy-to-noise ratio on the source-destination link.}
	\label{fig: irregular finite length}
\end{figure}

\begin{figure}[!hbtp]
   \centering
   {\includegraphics[width = 0.33 \textwidth, angle = -90]{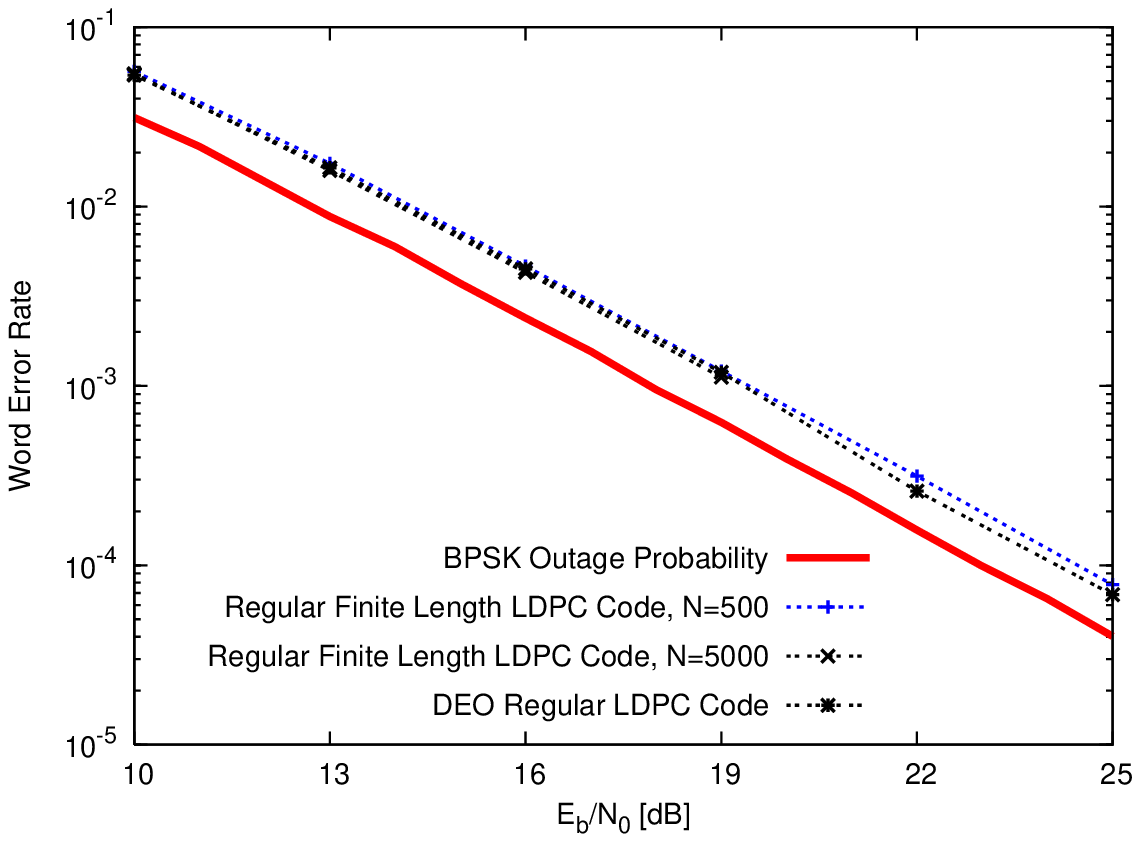}}
   \caption{Comparison of RCR-LDPC codes for different block lengths with iterative decoding on a cooperative MAC for two users, coding rate $R_c=1/3$. The ratio $E_b/N_0$ is the average information bit energy-to-noise ratio on the source-destination link.}
	\label{fig: finite_LDPC_1}
\end{figure}

\subsection{Comparison with Previous Work}
As mentioned in the introduction, especially rate-compatible punctured convolutional codes (RCPC) have been used in coded cooperation. The main drawback of these codes is that the WER increases with the logarithm  of the block length to the power $d$  where $d$ is the diversity  order \cite{bou2004tcd}, \cite{bou2005aoc}, whereas the WER of near-outage codes should  be independent of the block  length.  This can be seen clearly on Fig. \ref{fig: finite_LDPC_2}, where we show the WER of two rate-compatible non-recursive non-systematic (75,53,47) convolutional codes with block length 500 and 5000 respectively. We used the same channel conditions and coding rate as in scenario 1. 

We also compared with another protocol, Decode and Forward (DF), using near-outage LDPC codes for this protocol. Despite the fact that this implementation has near-outage performance, the WER performance is worse than that of our code construction. The reason is that the outage probability limit of DF is higher than that of coded cooperation.

\begin{figure}[!hbtp]
   \centering
   {\includegraphics[width = 0.33 \textwidth, angle = -90]{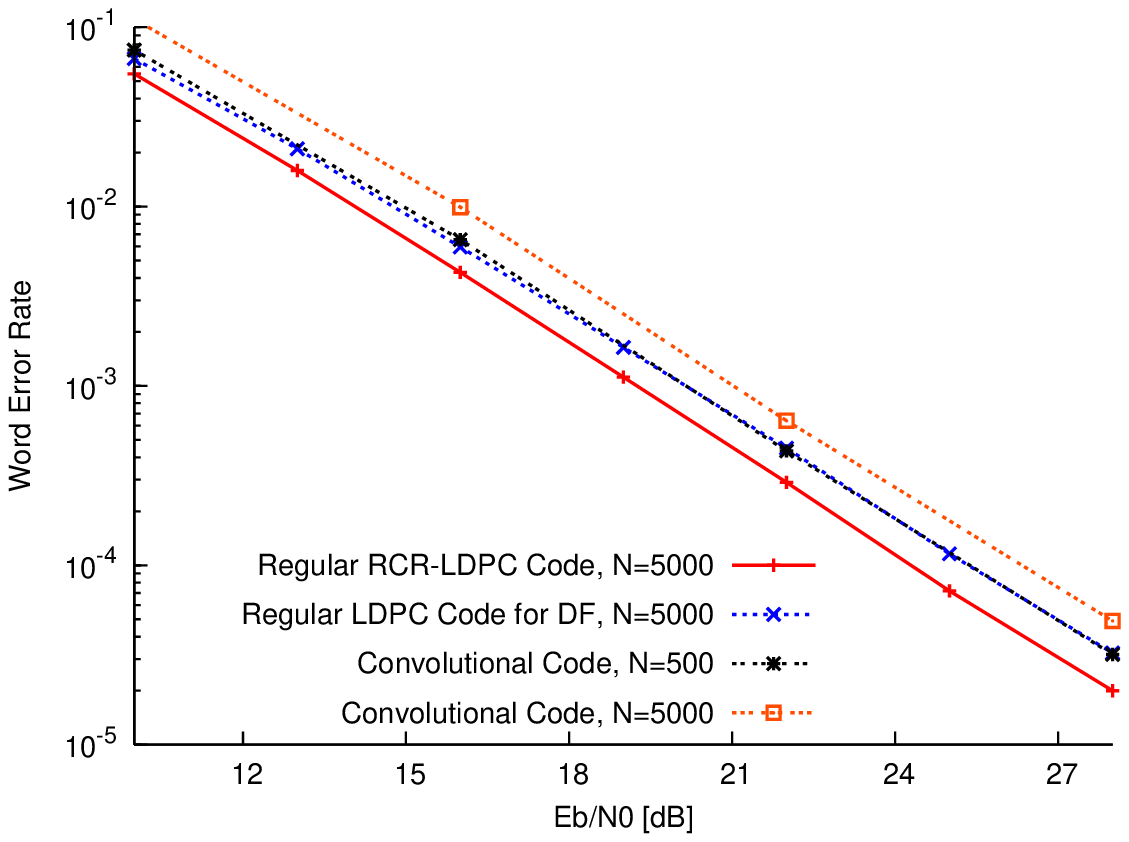}}
   \caption{Comparison of RCR-LDPC codes for coded cooperation with other work on a cooperative MAC for two users. We simulated LDPC codes for Decode and Forward under iterative decoding and an implementation of rate-compatible convolutional codes \protect\cite{hunter2004cc}. The ratio $E_b/N_0$ is the average information bit energy-to-noise ratio on the source-destination link.}
	\label{fig: finite_LDPC_2}
\end{figure}

\subsection{Comparison with fully random LDPC codes}

Finally, a comparison with random LDPC codes is made. In Sec. \ref{sec: Rate-compatible full-diversity LDPC codes}, the global parity-check matrix is obtained by  embedding  the root-LDPC  matrix  (Fig.  \ref{fig:  parity-check full-diversity})  into  the  rate-compatible  matrix  (Fig.  \ref{fig: parity1}). When using codes that are fully random generated, i.e., no special rootchecks are used, then the global parity-check matrix is obtained by  embedding a random LDPC  matrix into  the  rate-compatible  matrix  (Fig.  \ref{fig: parity1}), see Fig. \ref{fig: random}, where $H_1$ and $H_2$ are randomly generated. 

\begin{figure}[!ht]
   \centering
   {\includegraphics[width = 0.45 \textwidth]{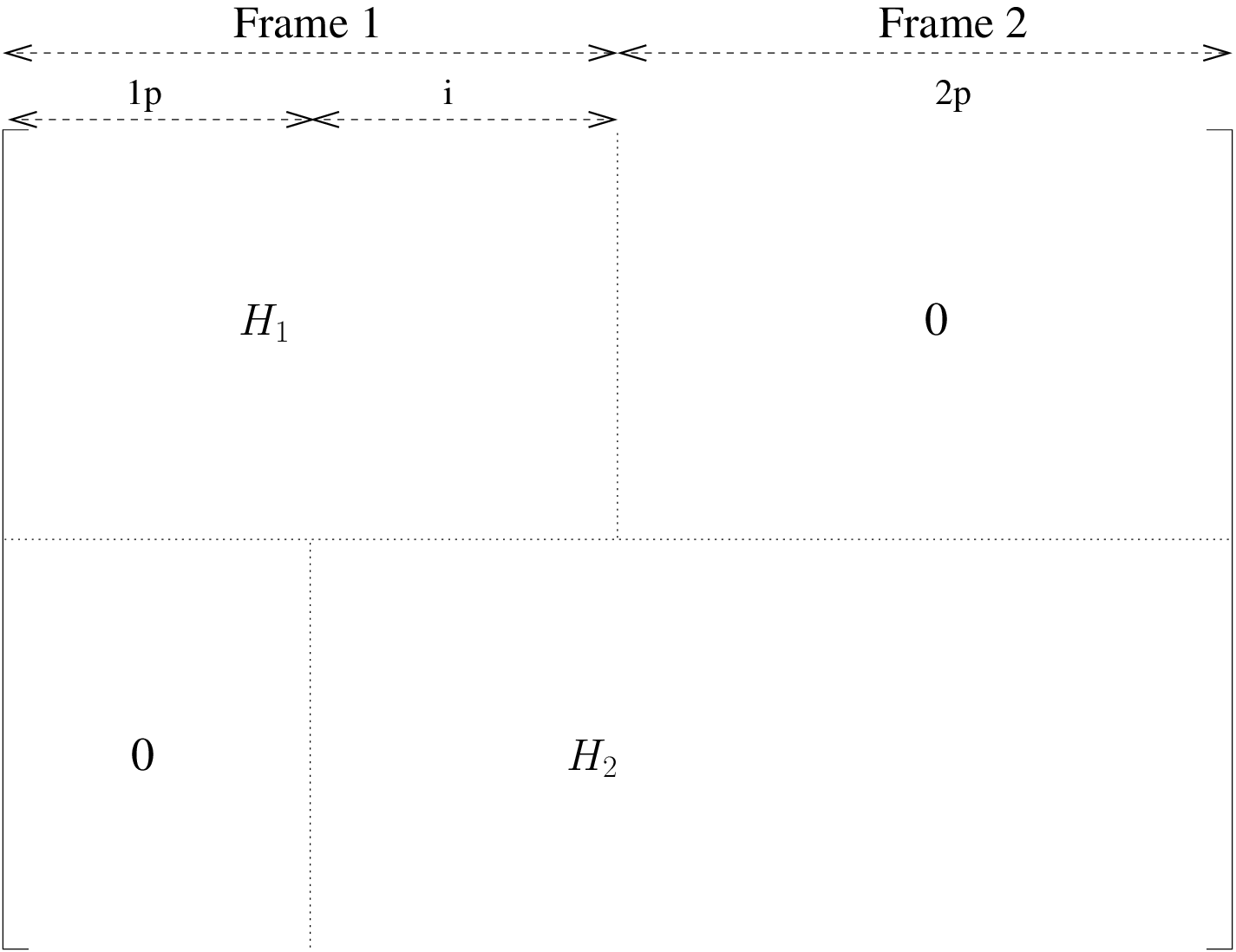}}
   \caption{Parity-check matrix of a rate-compatible LDPC code obtained by the extension of higher rate codes. Symbols are split into three classes: $i$ for the information bits, $1p$ and $2p$ for two classes of parity bits. The classes $i$ and $1p$ are transmitted by the source in frame 1. Parity bits $2p$ are transmitted in the second frame, for example by the relay after successful decoding of the first frame. Matrix $H_1$ is used to protect the information bits on the source channel. The parity bits generated by the relay provide an extra protection through the code $H_2$.}
	\label{fig: random}
\end{figure}

We simulated the same scenarios from the previous subsections, using the same code for $H_1$ and using the degree distribution of previously published excellent LDPC codes for the Gaussian channel for the random generation of $H_2$.

\subsubsection{Scenario 1}
{\small
\begin{eqnarray*}
	\lambda_2(x) &=& 0.189 x + 0.177 x^2 + 0.136 x^4 + 0.126 x^5 + 0.027 x^6 \\ 
	&& + 0.037 x^{11} +  0.006 x^{13} + 0.076 x^{21} + 0.225 x^{28}, \\
	\rho_2(x) &=& 0.153 x^4 +	0.125 x^5 +  0.040 x^6 	+ 0.261 x^7 \\
	&& + 0.149 x^8 + 0.178 x^9 + 0.041 x^{10} +	0.055 x^{11},
\end{eqnarray*}}
where the coding rate of $(\lambda_2(x),\rho_2(x))$ is $R_{c2}=0.4$, so that the overall coding rate is $R_c = 1/3$. The comparison with a regular $(3,9,3,6)$ RCR-LDPC code is shown in Fig. \ref{fig: finite_LDPC_random1}.

\begin{figure}[!htb]
   \centering
   {\includegraphics[width = 0.33 \textwidth, angle = -90]{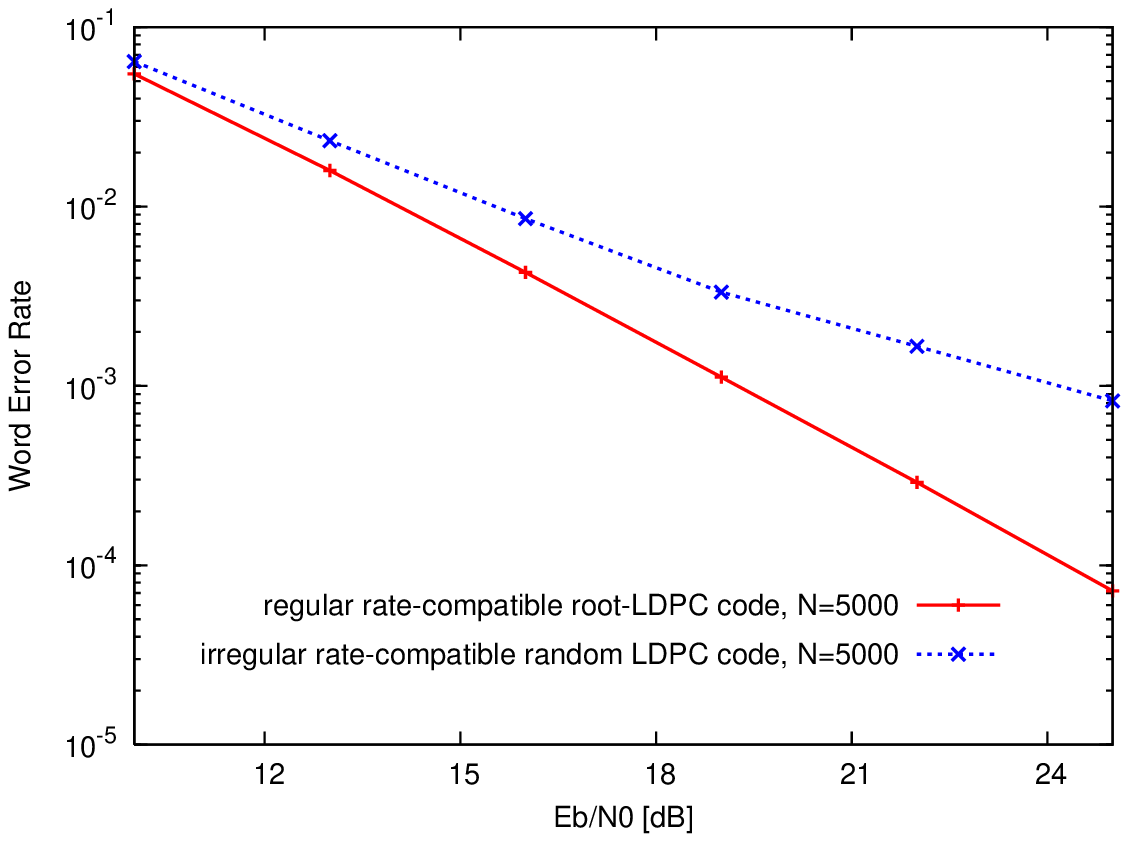 }}
   \caption{Comparison of RCR-LDPC codes with rate-compatible random LDPC codes for coded cooperation on a cooperative MAC for two users, coding rate $R_c=1/3$. The ratio $E_b/N_0$ is the average information bit energy-to-noise ratio on the source-destation link.}
	\label{fig: finite_LDPC_random1}
\end{figure}

\subsubsection{Scenario 2}
{\small
\begin{eqnarray*}
	\lambda_2(x) &=& 0.230 x + 0.164 x^2 + 0.149 x^5 +  0.126  x^6 + 0.027 x^7 \\
	&& + 0.037 x^{15} + 0.006 x^{16} + 0.243  x^{17} +  0.018  x^{23} , \\
	\rho_2(x) &=& 0.153  x^{5} + 0.425  x^{7} + 0.149 x^{8} + 0.273 x^{9}  ,
\end{eqnarray*}}
where the coding rate of $(\lambda_2(x),\rho_2(x))$ is $R_{c2}=9/19$, so that the overall coding rate is $R_c = 0.45$. The comparison with an irregular RCR-LDPC code is shown in Fig. \ref{fig: finite_LDPC_random2}.

\begin{figure}[!htb]
   \centering
   {\includegraphics[width = 0.33 \textwidth, angle = -90]{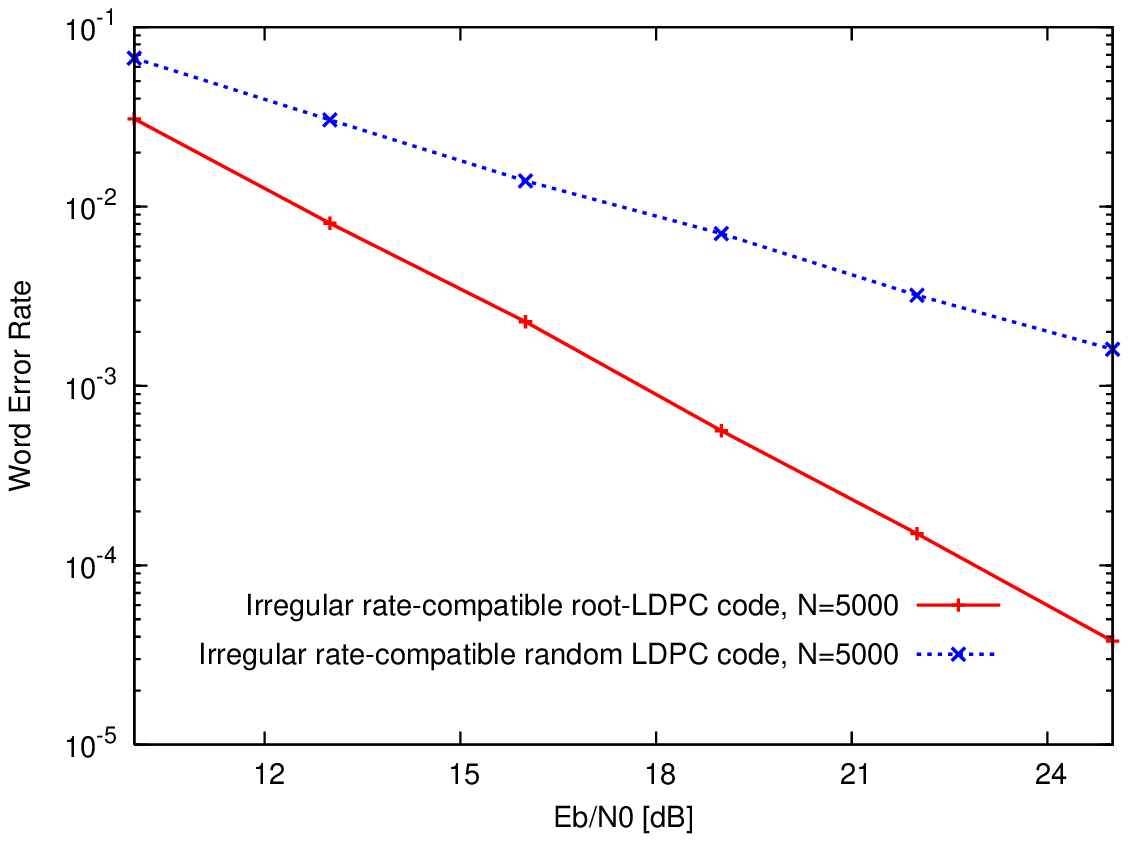 }}
   \caption{Comparison of RCR-LDPC codes with rate-compatible random LDPC codes for coded cooperation on a cooperative MAC for two users, coding rate $R_c=0.45$. The ratio $E_b/N_0$ is the average information bit energy-to-noise ratio on the source-destination link.}
	\label{fig: finite_LDPC_random2}
\end{figure}

In scenario 1, the threshold of $(\lambda_2(x),\rho_2(x))$ is $E_b/N_0 = 0.1$dB which is $0.338$dB from the Shannon limit; and in scenario 2, the threshold of $(\lambda_2(x),\rho_2(x))$ is $E_b/N_0 = 0.4$dB which is $0.33$dB from the Shannon limit. Despite the excellent thresholds of the codes in both scenarios, full-diversity is not achieved. From these two examples, it is clear that rootchecks are necessary to have full-diversity. 

%%------------------------------------------------------------------
%%------------------------------------------------------------------
\section{Conclusion}
We  have studied  LDPC codes  for relay  channels in  a slowly varying
fading environment  under iterative  decoding. We have  introduced the
new  family of  rate-compatible  root-LDPC codes,  which combines  the
rate-compatibility property  with the full-diversity  property for any
coding rate $R_c \leq R_{cmax} = \textrm{min}(\beta, 1-\beta)$, where $\beta$ is the
cooperation level.  Through a density evolution analysis and finite length simulations, we have shown  that the error rate  performance of
regular and irregular rate-compatible  root-LDPC codes is close to the outage probability  limit and this occurs for all block lengths  (finite and
infinite) and  all rates not exceeding $R_{cmax}$.  Its flexibility and  high performance makes rate-compatible  root-LDPC   attractive  for wireless cooperative communications scenarios with slowly varying fading.

%%------------------------------------------------------------------
%%------------------------------------------------------------------

\vspace{1.5 cm}

\noindent \textbf{Dieter Duyck} (S'09) received the M.S. degree in electrical engineering in 2007 from the Katholieke Universiteit Leuven (KUL), Leuven, Belgium. In 2006, he spent one year with the Communications and Electronics Department, at the Ecole Nationale Sup\'erieure des T\'el\'ecommunications (ENST, Telecom ParisTech), Paris, France. In 2007, he started his Ph.D. research at the Department of Telecommunications and Information Processing (TELIN), Ghent University, Gent, Belgium. 

From Oct. 2007 until present, he conducted his Ph.D. research. He has held visiting appointments with Ecole Nationale Sup\'erieure des T\'el\'ecommunications (ENST), Paris, France; and Texas A\&M University at Qatar, Doha, Qatar. His research interests are in communication theory, information theory, channel coding, joint network-channel coding, digital modulation and space-time coding.

M. Sc. Dieter Duyck received the first Young Researcher Award for \cite{duy2009fjncc} awarded by the Award Committee, formed by all Advisory Board Members of the European Newcom++ (Network of Excellence in Wireless COMmunications). He also received the best student paper award at the IEEE Symposium on Communications and Vehicular Technology in the Benelux (SCVT) in 2010. \\

\vspace{1.5 cm}

\noindent \textbf{Joseph Jean Boutros} (M'94, SM'09) received the M.S. degree in electrical engineering in 1992 and the Ph.D. degree in 1996, both from Ecole Nationale Sup\'erieure des T\'el\'ecommunications (ENST, Telecom ParisTech), Paris, France. From 1996 to 2006, he was with the Communications and Electronics Department, ENST, as an Associate Professor. He was also a member of the research unit UMR-5141 of the French National Scientific Research Center (CNRS). In 2007, he joined Texas A\&M University at Qatar (TAMUQ) as a full Professor in the electrical engineering program. He has been a scientific consultant for Alcatel Espace, Philips Research, and Motorola Semiconductors, and was a member of the Digital Signal Processing team of Juniper Networks Cable. His fields of interest are codes on graphs, iterative decoding, joint source-channel coding, space-time coding, and lattice sphere packings.\\

\vspace{1.5 cm}

\noindent \textbf{Marc Moeneclaey} (M'93, SM'99, F'02) received the diploma of electrical engineering and the Ph.D. degree in electrical engineering from Ghent University, Gent, Belgium, in 1978 and 1983, respectively. 

He is Professor at the Department of Telecommunications and Information Processing (TELIN), Gent University. His main research interests are in statistical communication theory, (iterative) estimation an detection, carrier and symbol synchronization, bandwidth-efficient modulation and coding, spread-spectrum, satellite and mobile communication. He is the author of more than 400 scientific papers in international journals and conference proceedings. Together with Prof. H. Meyr (RWTH Aachen) and Dr. S. Fechtel (Siemens AG), he co-authors the book Digital communication receivers – Synchronization, channel estimation, and signal processing. (J. Wiley, 1998). He is co-recipient of the Mannesmann Innovations Prize 2000.

During the period 1992-1994, he was Editor for Synchronization, for the IEEE Transactions on Communications. He served as co-guest editor for special issues of the Wireless Personal Communications Journal (on Equalization and Synchronization in Wireless Communications) and the IEEE Journal on Selected Areas in Communications (on Signal Synchronization in Digital Transmission Systems) in 1998 and 2001, respectively.

%%------------------------------------------------------------------
%%------------------------------------------------------------------
\end{document}